\documentclass[paper=a4, fontsize=11pt]{scrartcl}
\usepackage[T1]{fontenc}
\usepackage{fourier}
\usepackage[english]{babel}								 
\usepackage{amsmath,amsfonts,amsthm, amssymb} 
\usepackage{graphicx}	

\usepackage{float}
\usepackage[title,titletoc,page]{appendix}        
\usepackage{color}
\usepackage{cite}
\usepackage{caption}
\usepackage{subcaption}

\usepackage{algorithmic}
\usepackage[ruled,linesnumbered,noend]{algorithm2e}

\numberwithin{equation}{section}		
\numberwithin{figure}{section}			
\numberwithin{table}{section}				

\newtheorem{defi}{Definition}[section]
\newtheorem{lem}{Lemma}[section]
\newtheorem{thm}{Theorem}[section]

\newtheorem{eg}{Example}[section]

\newtheorem{prop}{Proposition}[section]

\title{IFF: A Super-resolution algorithm for Multiple Measurements}
\author{
Zetao Fei\thanks{\footnotesize Department of Mathematics, 
	HKUST,  Clear Water Bay, Kowloon, Hong Kong (zfei@connect.ust.hk).}\;
	and Hai Zhang\thanks{\footnotesize 
	Department of Mathematics, 
	HKUST,  Clear Water Bay, Kowloon, Hong Kong (haizhang@ust.hk). Hai Zhang was partially supported by Hong Kong RGC grant GRF 16304621.}}
	
\date{April 25, 2023}

\begin{document}
	\maketitle
	\begin{center}
		\textbf{Abstract}
	\end{center}
We consider the problem of reconstructing one-dimensional point sources from their Fourier measurements in a bounded interval $[-\Omega, \Omega]$. This problem is known to be challenging in the regime where the spacing of the sources is below the Rayleigh length $\frac{\pi}{\Omega}$. In this paper, we propose a super-resolution algorithm, called Iterative Focusing-localization and Filtering (IFF), to resolve closely spaced point sources from their multiple measurements that are obtained by using multiple unknown illumination patterns. The new proposed algorithm has a distinct feature in that it reconstructs the point sources one by one in an iterative manner and hence requires no prior information about the source numbers. The new feature also allows for a subsampling strategy that can  reconstruct sources using small-sized Hankel matrices and thus circumvent the computation of singular-value decomposition for large matrices as in the usual subspace methods. In addition, the algorithm can be paralleled.  A theoretical analysis of the methods behind the algorithm is also provided. The derived results imply a phase transition phenomenon in the reconstruction of source locations which is confirmed in the numerical experiment. Numerical results show that the algorithm can achieve a stable reconstruction for point sources with a minimum separation distance that is close to the theoretical limit. The efficiency and robustness of the algorithm have also been tested. This algorithm can be generalized to higher dimensions. 
\medskip
	
\textbf{Keywords}: super-resolution, multiple measurements, phase transition.  
	
\section{Introduction}\label{section:multiintroduction}
In optical microscopy, super-resolution techniques are dedicated to improving the resolution of optical signals. The resolution, however, is limited by the band-limited optical transfer function, which is a consequence of the optical diffraction limit. This resolution limit can be characterized using the so-called Rayleigh limit or Rayleigh wavelength which simply depends on the cutoff frequency in the signal. In recent years, super-resolution techniques, either physical or computational or combined, are developed to break the conventional Rayleigh limit.

In this paper, we are interested in developing computational methods or algorithms for the super-resolution problem which aims to reconstruct one-dimensional point sources from their Fourier
measurements in a bounded interval $[-\Omega, \Omega]$. We note that this problem is also closely related to the problem of Direction-of-Arrival (DOA) estimation. 
The first method for DOA can be dated back to the Bartlett method in 1948 \cite{10.1093/biomet/37.1-2.1} \cite{Bartlett1948}. Several important conventional methods are subsequently introduced to acquire higher resolution, such as MVDR \cite{1449208} and Linear Prediction \cite{1456696}. While easy to implement, these methods still provide limited accuracy and resolution. Subspace-based methods, such as MULtiple SIgnal Classification (MUSIC) \cite{1143830}, Estimation of Signal Parameters via Rotational Invariance Techniques (ESPRIT) \cite{32276}, and Matrix Pencil (MP) \cite{hua1990matrix}, have been developed to achieve higher resolution through eigen-decomposition. The idea of subspace methods dates back to Prony's method in 1795 \cite{Prony-1795}. We refer to \cite{batenkov2019super}\cite{LI2021118}\cite{li2019super}\cite{donoho1992superresolution}\cite{demanet2015recoverability}\cite{Moitra:2015:SEF:2746539.2746561}\cite{9410626}\cite{LIU2022402} for theoretical grounds for the stability of subspace method and the discussions on computational resolution limit.

In the last two decades, sparse exploiting methods have become popular, motivated by techniques in Compressive Sensing (CS) and sparse signal representation. These methods have broad applications, including scenarios with unknown source numbers and limited snapshots. Representative algorithms include LASSO, FOCUSS, TV norm minimization, Atomic norm minimization, SBL, SPICE, B-LASSO see \cite{doi:10.1137/0907087}\cite{10.2307/2346178}\cite{558475}\cite{738251}\cite{candes2014towards}\cite{candes2013super}\cite{tang2014near}\cite{chi2020harnessing}\cite{duval2015exact}\cite{AZAIS2015177}\cite{1315936}\cite{tipping2001sparse}\cite{4524050}\cite{5617289} etc. Theoretical results show that these methods can provably reconstruct off-the-grid point sources from noisy measurements under the assumption that the minimum separation distance between the point sources is above several Rayleigh wavelengths. A systematic introduction and review of DOA estimation techniques can be found in \cite{stoica2005spectral}\cite{godara2018handbook}\cite{YANG2018509} for example.

Intuitively, with multiple measurements that are obtained from multiple different illuminations, the accumulated information in the multiple measurements should contribute to a higher resolution. Indeed, subspace methods can take this advantage. In \cite{8374063}\cite{7547372}, a subspace method is applied to the aligned Hankel matrices to achieve better resolution. We also refer to \cite{liu2022mathematical} for the theoretical results. 
In this paper, we aim to propose a new super-resolution method/algorithm for multiple measurements. The method is motivated by the super-resolution techniques STED \cite{Hell2003-tt} and STORM \cite{Rust2006-kt}. We note that 
STED \cite{Hell2003-tt} achieves super-resolution by selectively deactivating fluorophores to minimize the illuminated area at the focal point, while STORM \cite{Rust2006-kt} breaks the Rayleigh limit by stochastically activating the individual photoactivatable fluorophores. The random one-by-one process inspires us to propose an iterative focusing-localization and filtering algorithm. We first numerically focus the illumination onto a single point source via an optimized linear combination of the available measurements, and then we localize its position. We next filter the signal of the localized point source from the measurements. We repeat this focusing and filtering process until all point sources are reconstructed. A theoretical analysis of the method behind the algorithm is also provided.

\subsection{Problem Setup}
In this paper, we consider the problem of reconstructing the locations of a collection of point sources from their multiple noisy measurements. Precisely, we consider the following mathematical model. Let $\mu=\sum_{j=1}^{n}a_{j}\delta_{y_j}$  be a discrete measure, 
where $y_j \in \mathbb R,j=1,\cdots,n$, represent the supports of the point sources and $a_j\in \mathbb C, j=1,\cdots,n$ their amplitudes. 

The point sources are illuminated by multiple illumination patterns $I_t$, $1\le t\le T$, and the illumination measures are
\[
\mu_t = \sum_{j=1}^n I_t(y_j)a_j\delta_{y_j}, \ 1\leq t \leq T.
\]
The measurements are the band-limited Fourier transform of the illuminated measures:
\begin{equation}\label{equ:multimodelsetting1}
 Y_t\left(\omega\right) = \mathcal{F} \mu_t \left(\omega\right) +  W_t\left(\omega\right)= \sum_{j=1}^{n}I_t(y_j)a_j e^{i y_j \omega} +  W_t\left(\omega\right), \quad 1\leq t\leq T, \ \omega \in [-\Omega, \Omega],
\end{equation}
where $\Omega$ is the cutoff frequency and $ W_t\left(\omega\right)$ is the noise. Throughout, we assume that $T\ge n$ i.e., the number of measurements is larger than the number of sources. With slight abuse of notations, we also denote $\mathcal{F} \mu_t$ as the function $\mathcal F \mu_t\left(\omega\right), \ \omega\in[-\Omega, \Omega]$ and $\|f\|_{\infty} = \sup_{\omega}|f\left(\omega\right)|$. We assume that $\|W_t\left(\omega\right)\|_{\infty}< \sigma$ with $\sigma$ being the noise level.
After uniformly sampling over the interval $[-\Omega, \Omega]$ at sample points $\omega_k = \frac{k}{K}\Omega, k = -K,\cdots, K$, we align all the measurements into a matrix as follows
\begin{align*}
\left(\begin{matrix}
Y_1\left(\omega_{-K}\right) & \cdots & Y_1\left(\omega_{K}\right)\\
\vdots &  &\vdots\\
Y_T\left(\omega_{-K}\right) & \cdots & Y_T\left(\omega_{K}\right)
\end{matrix}\right)
&=
\left(\begin{matrix}
I_1(y_1) & \cdots & I_1(y_n)\\
\vdots &  &\vdots\\
I_T(y_1) & \cdots & I_T(y_n)
\end{matrix}\right)
\left(\begin{matrix}
a_1 & & \\
 & \ddots &\\
 & & a_n\\
\end{matrix}\right)
\left(\begin{matrix}
e^{iy_1\omega_{-K}} & \cdots & e^{iy_1\omega_{K}}\\
\vdots &  &\vdots\\
e^{iy_n\omega_{-K}} & \cdots & e^{iy_n\omega_{K}}
\end{matrix}\right)\\
&\ +\left(\begin{matrix}
W_1\left(\omega_{-K}\right) & \cdots & W_1\left(\omega_{K}\right)\\
\vdots &  &\vdots\\
W_T\left(\omega_{-K}\right) & \cdots & W_T\left(\omega_{K}\right)
\end{matrix}\right).
\end{align*}
We denote the above equations as 
\begin{equation}{\label{image_eq}}
	Y = LAE+W.
\end{equation}
We observe that each row of $L$ represents an illumination pattern acting on different point sources while each column of $L$ represents different illuminations on a fixed point source.

Finally, for ease of notation, we denote for a given measure $\nu$,
$$
[\nu]=\left(\mathcal{F}\nu\left(\omega_{-K}\right),\mathcal{F}\nu\left(\omega_{-K+1}\right),\cdots,\mathcal{F}\nu\left(\omega_{K}\right)\right). 
$$

\subsection{Main contribution}
In this paper, we are dedicated to developing a new super-resolution algorithm to reconstruct off-the-grid point sources from multiple noisy measurements. We point out that we do not assume the noise pattern throughout. Unlike existing super-resolution algorithms that reconstruct all the point sources at once, the proposed method iteratively processes the following steps:
\begin{enumerate}
    \item Achieve numerical source focusing by solving an optimization problem.
    \item Localize the focused point source by using a subspace method.
    \item Filter out the recovered source by using  properly designed annihilating filters.
\end{enumerate}
We summarize the proposed algorithm as iterative focusing-and-filtering-based source identification. This paper also provides a theoretical analysis of the methods behind IFF. Through the discussion, we observe several features of IFF:
\begin{itemize}
    \item After the source focusing step, we only need to reconstruct a single source at each localization step. This fact allows us to deal with an optimization problem of rank-one Hankel matrices and thus allows the application of a subsampling strategy. 
    \item IFF is a prior-free super-resolution method of high accuracy.
    \item IFF can achieve stable reconstruction for point sources with a minimum separation distance that is close to the theoretical limit.
\end{itemize}
\subsection{Organization of the paper}
The paper is organized in the following way. In Section \ref{sec:source focusing} and Section \ref{sec:source removal}, we explain the main ideas of source focusing and source filtering with theoretical discussion. The detailed implementation of the proposed algorithm can be found in Section \ref{sec:implementation}. In Section \ref{sec:numerical experiment}, we present numerical results of the IFF Method. We conclude the paper with a discussion on future research topics in Section \ref{sec:conclusion}.

\section{Source Focusing and Localization}{\label{sec:source focusing}}
This section is dedicated to the first part of the proposed algorithm, namely source focusing and localization. We begin with introducing the idea of the source focusing. In Section \ref{subsec:dis_snr_focus}, we discuss the change of SNR resulting from source focusing. The design of the source focusing algorithm is proposed in Section \ref{subsec:algo_focus}. Finally, a theoretical analysis of the method behind is given in Section \ref{subsec:mini_separation and error}.

We first introduce some concepts and notations that are used throughout the paper.
We assume the illumination matrix, $L$, has linearly independent columns and we write it into column blocks as follows:
\begin{equation*}
	L = \left(
	\begin{matrix}
	\alpha_1,\alpha_2,\cdots,\alpha_n
	\end{matrix}\right).
\end{equation*}
For $j = 1,\cdots,n$, we define matrix $L_{j}$ as
\begin{equation*}
	L_{j} = \left(
	\begin{matrix}
	\alpha_1,\cdots,\alpha_{j-1},\alpha_{j+1},\cdots,\alpha_n
	\end{matrix}\right).
\end{equation*}
We denote the projection map onto the column space of $L_{j}$ as $\mathcal{P}_{L_{j}}$ and the identity map as $\mathcal{I}$. From basic linear algebra, we know that $\mathcal{P}_{L_{j}} = L_{j}\left(L_{j}^*L_{j}\right)^{-1}L_{j}^*$.

We now consider the task of focusing the illumination onto the $j$-th source $y_j$ by using linear combinations of the given measurements. For this purpose, we denote by $U_j$ the permutation matrix that permutes the $j$-th column with the $n$-th column. We rewrite (\ref{image_eq}) as follows:
\begin{align}{\label{permutation_Y}}
    Y=LU_j\cdot U_j^{-1}AE+W.
\end{align}
Applying the QR decomposition to $LU_j$, we have
\begin{align}{\label{QR decomp}}
    LU_j = Q\left(\begin{matrix}
	R\\0
	\end{matrix}\right).
\end{align}
Multiplying $Q^{*}$ on the both sides of (\ref{permutation_Y}) yields
\begin{equation}{\label{QtY}} 
	Q^{*}Y = \left(\begin{matrix}
	R\\0
	\end{matrix}\right)U_j^{-1}AE+Q^{*}W.
\end{equation}
We observe that the n-th row of $Q^{*}Y$ can then be regarded as the measurement generated by the single source $y_j$. Indeed, the $n$-th row of equation (\ref{QtY}), denoted as $\tilde{Y}_j$, gives
\begin{align}{\label{tilde_Y_j}}
\tilde{Y}_j 
&= \left(\sum_{t=1}^{T}q_{tn}Y_t\left(\omega_{-K}\right),\cdots,\sum_{t=1}^{T}q_{tn}Y_t\left(\omega_{K}\right)\right) \notag\\
&=R_{nn}a_j\left(e^{iy_j\omega_{-K}},\cdots,e^{iy_j\omega_{-K}}\right)+\left(\sum_{t=1}^{T}q_{tn}W_t\left(\omega_{-K}\right),\cdots,\sum_{t=1}^{T}q_{tn}W_t\left(\omega_{K}\right)\right)\notag\\
&\triangleq \tilde{a}_j\left(e^{iy_j\omega_{-K}},\cdots,e^{iy_j\omega_{-K}}\right)+\tilde{W}_j.
\end{align}
We note that the component $\tilde{a}_j$, which depends on the permutation matrix $U_j$, gives an effective illumination amplitude on the $j$-th source while all the other sources are ''quenched'' by the linear combination of given illumination patterns.

We can use the measurement from $\tilde{Y}_j $ to reconstruct the position of focused point source $y_j$ using a standard subspace method, such as  MUltiple SIgnal Classification (MUSIC) and Matrix Pencil (MP). For this purpose, we form the following Hankel matrix for $\tilde{Y}_j$, 
\begin{align}
\tilde{H_j} = \left(
\begin{matrix}
\tilde{Y}_j\left(\omega_{-K}\right) & \tilde{Y}_j\left(\omega_{-K+1}\right) & \cdots & \tilde{Y}_j\left(\omega_{0}\right)\\
\tilde{Y}_j\left(\omega_{-K+1}\right) & \tilde{Y}_j\left(\omega_{-K+2}\right) & \cdots & \tilde{Y}_j\left(\omega_{1}\right)\\
\vdots & \vdots& &\vdots\\
\tilde{Y}_j\left(\omega_{0}\right) &\tilde{Y}_j\left(\omega_{1}\right)& \cdots & \tilde{Y}_j\left(\omega_{K}\right)
\end{matrix}
\right).
\end{align}
From (\ref{tilde_Y_j}), we see that $\tilde{H_j}$ is a linear combination of $\{H_t\}_{t=1}^T$ with coefficients $\{q_{tn}\}_{t=1}^T$, i.e.,
\begin{equation}  \label{eq-linearcomb}
    \tilde{H_j} = \sum_{t=1}^T q_{tn} H_t. 
\end{equation}
Here, the matrices $H_t$'s are the Hankel matrices associated with the given measurements $Y_t$'s in (\ref{equ:multimodelsetting1}), i.e. 
\begin{align}{\label{hankel}}
H_t = \left(
\begin{matrix}
Y_t\left(\omega_{-K}\right) & Y_t\left(\omega_{-K+1}\right) & \cdots & Y_t\left(\omega_{0}\right)\\
Y_t\left(\omega_{-K+1}\right) & Y_t\left(\omega_{-K+2}\right) & \cdots & Y_t\left(\omega_{1}\right)\\
\vdots & \vdots& &\vdots\\
Y_t\left(\omega_{0}\right) &Y_t\left(\omega_{1}\right)& \cdots & Y_t\left(\omega_{K}\right)
\end{matrix}
\right),~~1\le t \le T.
\end{align}

It is important to notice that $\tilde{H_j}$ is the summation of a rank-one matrix and a noise matrix. We shall exploit this fact in our numerical reconstruction of the position of the focused point source. On the other hand, since the signal-to-noise ratio (SNR) plays a crucial role in the reconstruction of the subspace methods, we provide 
a theoretical analysis of the SNR for the focused measurement $\tilde{Y_j}$ 
in the next subsection. 

\subsection{Discussion on SNR for the measurement of perfectly focusing}{\label{subsec:dis_snr_focus}}
We start with the following estimation for each component of $\tilde{W}_j$ in equation (\ref{tilde_Y_j}):
\begin{align}{\label{noise level after focusing}}
    |\tilde{w}_{ij}| = |\sum_{n=1}^{T} q_{ni}W_{nj}|\le \sqrt{\sum_{n=1}^{T}q_{ni}^2\cdot\sum_{n=1}^{T}W_{nj}^2}\le \sqrt{T}\sigma.
\end{align}
We then consider $\tilde{a}_j$ in (\ref{tilde_Y_j}). It is clear that it depends on $R_{nn}$.
We can calculate that 
\begin{align}{\label{Rnn}}
    R_{nn} = \|\left(\mathcal{I}-\mathcal{P}_{L_j}\right)\alpha_j\|_2.
\end{align}
Intuitively, the smaller $R_{nn}$ is, the lower SNR is and the more difficult the super-resolution problem is. The relationship between $R_{nn}$ and the illumination matrix $L$ can be illustrated through the following simple examples. First, we assume $L$ has a deficient row rank. In this case, we have $R_{nn}=0$. Without sufficient valid measurements, it is impossible to numerically focus on a single source. It also implies that adding the same measurement is not helpful since it cannot affect $R_{nn}$. Meanwhile, when $L$ has the same row vectors, the problem will degenerate to the single snapshot case.
Then, let us consider the case when $L$ is diagonal. It is an ideal case where the sources are automatically focused in each illumination. In this case, $R_{nn}$ equals the illumination amplitude for the illumination onto the source $y_j$.

We now derive some quantitative characterization of $R_{nn}$. Under the assumption that the illumination matrix consists of independent and identically distributed random variables with mean $u$ and variance $v^2$, we have the following estimation for the second-order moment:
\begin{prop}{\label{L_eff u not 0}}
    For $R_{nn}$ given in (\ref{Rnn}), we have
    \begin{align}
        \mathbb{E}\ R_{nn}^2 \ge \left(T-n+1\right)v^2.
    \end{align}
\end{prop}
Especially, when $u=0$ and the components of $L$ are i.i.d. subgaussian random variables, we have the following concentration property for $R_{nn}$:
\begin{prop}{\label{L_eff}}
	For $u = 0$ and $\|L_ {ij}\|_{\psi_2}\le B$, for any $t>0$, we have
	\begin{equation}{\label{Leff_lowerbd}}
		\mathbb{P} \left(|R_{nn}^2-\left(T-n+1\right)v^2|>t\right) \le 2\exp{\left(-c\min\left(\frac{t^2}{B^4\left(T-n+1\right)},\frac{t}{B^2}\right)\right)},
	\end{equation}
	where $c$ is some positive absolute constant and $\|\cdot\|_{\psi_2}$ is the subgaussian norm defined as 
	\begin{equation}\label{psi2_cond}
    \|L_{ij}\|_{\psi_2} = \sup_{p\ge1} p^{-1/2}\left(\mathbb{E}L_{ij}^p\right)^{1/p}.
	\end{equation}
\end{prop}
In practice, the illumination amplitudes can be assumed to be bounded, i.e. the range of $L_{ij}$ is bounded. Under this assumption, the condition expressed by (\ref{psi2_cond}) is always met. As for the order of $R_{nn}$, we let $t = \frac{v^2}{2}\left(T-n+1\right)$, then the inequality (\ref{Leff_lowerbd}) gives 
that $R_{nn}\sim \mathcal{O}\left(\sqrt{T}\right)$ with probability $1-2e^{-\mathcal{O}\left(T\right)}$. Combining this result with (\ref{noise level after focusing}), we observe that under the assumption in Proposition \ref{L_eff}, numerical focusing via linear combination can preserve SNR with high probability if we use enough linearly independent measurements.

\subsection{Algorithm for focusing}{\label{subsec:algo_focus}}

For the reader's convenience, we review the model and the goal at the beginning of this section. For given point sources $\mu$, we have $T$ noisy measurements $\{Y_t\}_{t=1}^T$ with different illumination pattern. The goal is to reconstruct the position of the point sources using a one-by-one strategy, which requires that an algorithm for sources focusing should be developed.

We have shown in the previous discussion that with sufficiently many linearly independent illumination patterns, we can achieve a perfect illumination pattern that focuses on a single point source through a proper linear combination of the given measurements, see (\ref{eq-linearcomb}).  
We now address the issue of finding the linear combination coefficients $\{q_{tn}\}$'s numerically.

It is clear that a constant multiplier of $\tilde{Y}_j$ does not affect the SNR.  In finding the linear combination coefficients, we do not need to apply additional constraints about the $\ell_2$-norm of $\left(q_{1,n},\cdots,q_{T,n}\right)$ though it is required to be one in QR decomposition. Let $N = \tilde{H}_j^*\tilde{H_j}$, we consider the following unconstrained optimization problem,
\begin{align}{\label{opt_pro}}
\min_{q_{tn}} f\left(q_{1,n},\cdots,q_{Tn}\right):=\frac{Tr\left(N\right)^2}{Tr\left(N^*N\right)}.
\end{align}
We observe that $f\ge 1$ and $f=1$ if and only if $\tilde{H_j}$ is a rank-$1$ matrix. We point out that the objective function $f$ is determined only by the ratios of singular values. By maximizing the gap between the largest singular value and the rest, $f$ can reflect the low-rank nature better than traditional objective functionals which use the nuclear norm or other norms. Meanwhile, in the computation of $f$, there is no need for singular value decomposition.
For the noiseless case, we observe that the solution to the optimization problem (\ref{opt_pro}) is not unique and each solution corresponds to a perfect focus on one of the point sources. In practice, with the presence of noise, we expect a solution to the optimization problem (\ref{opt_pro}) can yield a good illumination pattern that is focused on a single point source in the sense that the contrast of the illumination amplitude on the focused point source and the others are sufficiently large so that the focused source can be distinguished and localized. The following estimate sheds light on this issue. 
\begin{prop}{\label{contrast}}
    Let $\mu=\sum_{j=1}^n a_j\delta_{y_j}$, $Y=\mathcal{F}\mu\left(\omega\right)+W\left(\omega\right)$, where $\omega\in[-\Omega,\Omega]$ and $\|W\|_\infty<\sigma$. Let $H$ be the Hankel matrix generated by $Y$ and $N=H^*H$ be the normal matrix of $H$. Assume $|a_s| = \max_{j} |a_j|$, and $|a_s|>\sum_{t\ne s} |a_t|+\sigma$, then we have the following estimation:
    \begin{align}
        \frac{Tr\left(N\right)^2}{Tr\left(N^*N\right)}\le \left(\frac{|a_s|+\sum_{t\ne s} |a_t|+\sigma}{|a_s|-\sum_{t\ne s} |a_t|-\sigma}\right)^4
    \end{align}
\end{prop}


In the next proposition, we derive an upper bound of the minimization problem (\ref{opt_pro}) when there is a single source with amplitude $M$ and noise level $\sigma$. 

\begin{prop}{\label{threshold}}
    Let $\mu=M\delta_y$, $Y=\mathcal{F}\mu\left(\omega\right)+W\left(\omega\right)$, where $\omega\in[-\Omega,\Omega]$, $\|W\|_\infty<\sigma$ and $M>2\sigma$. Let $H$ be the Hankel matrix generated by $Y$ and $N=H^*H$ be the normal matrix of $H$, then we have the following estimation:
    \begin{align}
        \frac{Tr\left(N\right)^2}{Tr\left(N^*N\right)}\le \left(1+4K\frac{\sigma^2}{M^2}\right)^2. 
    \end{align}
\end{prop}

The upper bound in the proposition above can serve as a theoretical ground for the threshold we propose in the following. Note that $\sigma$ therein is the SNR after the focusing step which is not available in practice. 
We shall, nevertheless, use the original SNR as a good guess for it (see the discussion at the end of the previous subsection). 
We therefore set
\begin{equation} \label{eq-gamma}
\Gamma=\left(1+4K\left(\frac{1}{SNR}\right)^2\right)^2
\end{equation}
as the threshold, where SNR represents the signal-to-noise ratio in the original problem. We reject all the solutions having function values larger than $\Gamma$. We call this step \textbf{clean-up step}. 
Notice that the solutions to the optimization problem (\ref{opt_pro}) depend on the choice of initial guesses. With the equation (\ref{permutation_Y}) in our mind, we expect that the solutions may give us several clusters of positions. We then take an average over each cluster to get the recovered point sources $\{\hat{y}_p\}_{p=1}^{P}$. Clearly, $P\le n$ here. The pseudo-code is given in \textbf{Algorithm \ref{algo_1}} below.

\begin{algorithm}
	\caption{Source Focusing and Localization}
	\label{algo_1}
    \DontPrintSemicolon
    \SetKwInOut{Input}{Input}
    \SetKwInOut{Output}{Output}
    \SetKwInOut{Init}{Initialization}
	\Input{Hankel matrices: $H_t$, $1\le t\le T$, Noise level: $\sigma$, Sampling number: $K$.}
	\Input{Tolerance rate: $\epsilon$.}
	\Init{Initial matrix: $I_{T\times T}$.}
		\For {$j\ =\ 1:T$}{
		Initial data: $q^{\left(0\right)}[:,j] = I[:,j]$.\;
		Solve (\ref{opt_pro}) to obtain $\{q_{tn}\}$ and $f\left(q_{1,n},\cdots,q_{Tn}\right)$ until $f\left(q_{1,n},\cdots,q_{Tn}\right)<1+\epsilon$. Get the result $q^{\left(r\right)}[:,j]$ and corresponding $\tilde{H}_j$.\;
		Apply MUSIC to $\tilde{H}_j$, get source estimation $\tilde{y}_j$.\;
        }
		Apply the clean-up step to the estimations and divide the results into several clusters. Take the average over each cluster to get the reconstructed sources $\{\hat{y}_j\}_{p=1}^P$.\;	
	\Return $S = \{\hat{y}_j\}_{p=1}^P$.\;
\end{algorithm}

We note that in \textbf{Algorithm \ref{algo_1}}, $\Gamma$ is only used in (\ref{eq-gamma}) for the clean-up step but not the tolerance rate for the optimization problem. This is because the accuracy of the solution to (\ref{opt_pro}) can affect the accuracy of the recovered position of the selected source, while $\Gamma$ is derived in the "worst-case scenario". In practice, we need to choose a tolerance rate that is small enough. We believe that a more detailed analysis of the objective function $f$ may lead to a better choice of tolerance rate. 
We also notice that in Step 3, $\tilde{H_j}$ is the sum of a rank-one matrix and a noise matrix, which allows us to form the matrix $\tilde{H_j}$ with a small size to reconstruct the source position. This motivates us to develop a subsampling strategy, which will be discussed in Section \ref{sec:implementation}. In Section \ref{subsec:snr_filter}, we will see another advantage of the subsampling strategy.

\subsection{Analysis of the  the numerical focusing and localization method}{\label{subsec:mini_separation and error}}
In this subsection, we analyze the reconstruction error of the numerical source focusing and localization method. We aim to address the following two issues: the first is under what conditions on the ground truth sources one step of source focusing and localization can reconstruct a source that is close to one of the ground truth sources. The second is the reconstruction error. It will be shown that the first issue is closely related to the minimum separation distance between the ground truth sources, which also sets a theoretical limit for the method. 

To proceed, we assume the ground truth point sources are supported at $\{y_j\}_{j=1}^n$. We denote the source obtained from the numerical focusing and localization step as 
$\mu = M\delta_y$ with $M>0$. We 
write the numerically focused measurement as
$
\tilde Y = [\mu] + W',
$
where $W' := \tilde Y - [\mu]$. We assume that 
$$
\| W'\|_{\infty} < \sigma' 
$$ 
for some $\sigma' >0$.  In view of (\ref{tilde_Y_j}), we see that in the ideal focusing case, $M=|\tilde{a}_j|$, $\sigma'\leq \sqrt{T}\sigma$. 
For convenience, we introduce the following definition. 

\begin{defi}{\label{sigma-admissible}}
Given a source $\mu$, we say a collection of point sources $\{\delta_{y_j}\}_{j=1}^n$ is $\sigma$-admissible to $\mu$ if there exist complex numbers $\hat{a}_j, j=1, 2, \cdots n$, such that $\hat{\mu}=\sum_{j=1}^n \hat{a}_j\delta_{y_j}$ satisfies 
\[
\|[\hat{\mu}]-[\mu]\|_{\infty}<\sigma.
\]
\end{defi}

We have the following theorem which shows that under certain minimum separation conditions to the ground truth sources, the focused source $\mu$ is close to one of the ground truth sources.

\begin{thm}{\label{separation_dist_upper}}
Let $n\ge 2$,  a collection of point sources $\{\delta_{y_j}\}_{j=1}^n$ is supported on $[-\frac{\pi}{2\Omega},\frac{\pi}{2\Omega}]$ satisfying the following condition:
\begin{align}\label{separation distance}
    \tau=\min_{p\ne q}|y_p-y_q| \ge \frac{3.03\pi e}{\Omega}\left(\frac{\sigma'}{M}\right)^{\frac{1}{n}}.
\end{align}
If $\{\delta_{y_j}\}_{j=1}^n$ is $\sigma'$-admissible to $\mu = M\delta_y$,
then 
\[
\min_{1\le j\le n} |y-y_j|<\frac{\tau}{2}.
\]
\end{thm}
 
On the other hand, the proposition below shows that if the separation distance of the ground truth sources is below a certain threshold, then it is not guaranteed that the numerical-focused source is close to any of the ground truth sources. 

\begin{prop}{\label{separation_dist_lower}}
	For given $0<\sigma'<M$, and integer $n\ge2$, let
\begin{align}{\label{tau}}
	\tau = \frac{0.96e^{-\frac{3}{2}}}{\Omega}\left(\frac{\sigma'}{M}\right)^{\frac{1}{n}}.
\end{align}
	For uniformly separated point sources $\{\delta_{y_j}\}_{j=1}^n$ with distance $\tau$. There exist $y_k\in\{y_j\}_{j=1}^n$ such that $\mu=M\delta_k$, $\hat{\mu}=\sum_{j\ne k}\hat{a}_j\delta_{y_j}$ satisfying $\|[\mu]-[\hat{\mu}]\|_{\infty}<\sigma'$.
\end{prop}

Finally, we consider the reconstruction error. Under the assumption that perfect focusing is achieved (see  (\ref{tilde_Y_j})), we have the following estimation.
\begin{thm}{\label{error_upper}}
Let $\mu=M\delta_y$ be the focused point source supported on $[-\frac{\pi}{2\Omega},\frac{\pi}{2\Omega}]$. If $\{\delta_{y_k}\}$ is $\sigma'$-admissible to $\mu$ for some $k$, then 
\begin{align}\label{error_est}
    |y-y_k|< \frac{\pi}{\Omega}\left(\frac{\sigma'}{M}\right). 
\end{align}
\end{thm}

\section{Annihilating Filter Based Source Filtering}{\label{sec:source removal}}
In this section, we propose an algorithm to remove the recovered point source/sources, which paves the way for a complete reconstruction. 

\subsection{Source Filtering Algorithm}
Literature usually uses the annihilating filters to achieve signal reconstruction, see \cite{1003065}
\cite{7547372}\cite{doi:10.1137/15M1042280}. In contrast, we apply annihilating filters to remove the reconstructed point source/sources. 
To illustrate the idea, we first look at a toy example for the signal processed by an annihilating filter. 
\begin{eg}
	Suppose we evenly sample a signal generated by a single point source with position $z$ and amplitude $a$, following the setup in section \ref{sec:source focusing}. We have
	\begin{align}
	Y = \left(ae^{iz\omega_{-K}},ae^{iz\omega_{-K+1}},\cdots,ae^{iz\omega_{K}}\right).
	\end{align}
	We then define a filter, $F$, by $F = \left(1,-e^{iz\frac{\Omega}{k}}\right)$. Calculating the discrete convolution of $Y$ and $F$, we derive
	\begin{align}
	Y\ast F=\left(ae^{-iz\Omega},0,0,\cdots,0,-ae^{iz\frac{K+1}{K}\Omega}\right).
	\end{align}
\end{eg}
We observe that though there is a boundary effect due to discrete convolution, the middle part of $Y\ast F$ shows the complete filtering of the source. Based on this observation, we design the following annihilating filter for the original measurements:
\begin{align}
	F = \left(1,-e^{i\hat{y}_1\frac{\Omega}{K}}\right)\ast\left(1,-e^{i\hat{y}_2\frac{\Omega}{K}}\right)\ast\cdots\ast\left(1,-e^{i\hat{y}_P\frac{\Omega}{K}}\right),
\end{align}
where the convolution is in the discrete sense and $\{\hat y_p\}_{p=1}^P$ are derived in \textbf{Algorithm \ref{algo_1}}.
We denote $|\alpha|$ the length of the vector $\alpha$ and $\alpha(r:s)$ the vector that extracts the $r$-th to the $s$-th elements of $\alpha$. Notice that $|F|=P+1$ and $|Y\ast F|=2K+P+1$, the boundary effect occurs at the first $P$ elements and also the last $P$ elements of $Y_t\ast F$. We, therefore, pick the middle part as the processed measurements with the explicit form 
$$
Y_t'= \left(Y_t\ast F\right)\left(P+1:2K+1\right), \quad t=1,\cdots,T.
$$
The pseudo-code is given in 
\textbf{Algorithm \ref{algo_2}} below. 
\begin{algorithm}
	\caption{Source Filtering}
	\label{algo_2}
	\DontPrintSemicolon
    \SetKwInOut{Input}{Input}
    \SetKwInOut{Output}{Output}
    \SetKwInOut{Init}{Initialization}
	\Input{Original measurements: $\{Y_t\}_{t=1}^T$, Recovered sources: $S =\{\hat{y}_p\}_{p=1}^P$, sampling distance $h = \frac{\Omega}{K}$.}
		Form the filter $F=\left(1,-e^{i\hat{y}_1h}\right)\ast\cdots\ast\left(1,-e^{i\hat{y}_Ph}\right)$.\;
		\For {$t\ =\ 1:T$}{
		Compute the discrete convolution $Y_t\ast F$\;
		Let $Y'_t=Y_t\ast F\left(P+1:2K+1\right)$.\;
		Let $s = \text{arg}\max_m\{s=2m-1,s\le 2K-|S|+1\}$, form $s\times s$ Hankel matrices $H'_t$ using $Y_t'\left(1:2s-1\right)$.\;
		}
	\Return $\{Y_t'\}_{t=1}^T$, $\{H'_t\}_{t=1}^T$.\;
\end{algorithm}

\subsection{Analysis on SNR}
{\label{subsec:snr_filter}}

In this subsection, we estimate the SNR of the signal obtained from the source filtering algorithm introduced in the previous subsection. The estimate will shed light on the reconstruction error for the next step. 
We first define $A_j$ by
\begin{align}
A_j = \prod_{p=1}^P \left(e^{iy_j\frac{\Omega}{K}}-e^{i\hat{y}_p\frac{\Omega}{K}}\right),\ j=1,2,\cdots,n,  
\end{align}
and $Q\left(x\right)$ by
\begin{align}
	Q\left(x\right) &= \prod_{p=1}^{P}\left(x-e^{i\hat{y}_p\frac{\Omega}{K}}\right) \triangleq c_0+c_1x+\cdots+c_Px^P.
\end{align}
To write $Y'$ more explicitly, we have
\begin{align}{\label{Y't}}
	Y_t'\left(s\right) = \sum_{j=1}^n A_ja_jI_t(y_j)e^{-iy_j\left(1-\frac{s}{K}\right)\Omega}+W'\left(s\right),\ s=0,1,\cdots,2K-P.  
\end{align}
From (\ref{Y't}) we see that both $\{A_j\}$ and $W'$ affect the SNR after filtering, which leads us to analyze $\{A_j\}$ and $W'\left(s\right)$ separately.

For $W'\left(s\right)$, we have the following worst-case estimation
\begin{align}{\label{W'_est}}
|W'\left(s\right)| &= |\sum_{l=0}^P c_{l}W\left(s+l\right)| 
\le \sigma \sum_{l=0}^{P} |c_l| 
=\sigma \sum_{l=0}^{P} |
\sum_{\substack{S\subset \{1,\cdots,P\}\\ |S|=l}}
\prod_{p\in S} \left(-e^{i\hat{y}_p\frac{\Omega}{K}}\right) | \notag\\
&\le \sigma \sum_{l=0}^{P} C_P^l=2^P\sigma.
\end{align}

For ${A_j}$, since $\{y_j\}$ are closely spaced and the sampling distance $\frac{\Omega}{K}$ is also small, we have the following estimation:
\begin{align}{\label{Aj_est}}
|A_j| &=\prod_{p=1}^P |e^{iy_j\frac{\Omega}{K}}-e^{i\hat{y}_p\frac{\Omega}{K}}| \notag\\
      &=\prod_{p=1}^{P}\sqrt{2-2\cos \left(\left(y_j-\hat{y}_p\right)\frac{\Omega}{K}\right)}\notag \\
      &\sim \left(\frac{\Omega}{K}\right)^P\prod_{p=1}|y_j-\hat{y}_p|.
\end{align}
We observe that in the ideal case, $y_j \in \{\hat{y}_p\}_{p=1}^P$, we will have $A_j=0$ which implies that the source positioned at $y_j$ is completely filtered out. Therefore, the accuracy of the source localization in \textbf{Algorithm \ref{algo_1}} plays an important role in the source filtering step. Inaccurate source filtering may cause additional errors.
The results in (\ref{W'_est}) and (\ref{Aj_est}) imply that the number of recovered sources and sampling distance have effects on SNR after filtering. 
We also notice that from (\ref{Aj_est}), the subsampling strategy can reduce the effect due to the increased value of $\frac{\Omega}{K}$. In addition, the source filtering step reduces the number of point sources to be reconstructed in the next step. Both ensure that the iterative strategy of IFF to reconstruct all the point sources is achievable.

\section{Implementation and Discussion}{\label{sec:implementation}}

In this section, we present the implementation of the IFF algorithm, which combines \textbf{Algorithm \ref{algo_1}} and \textbf{Algorithm \ref{algo_2}}. Meanwhile, we point out a rule for the choice of stopping criterion in the iterations and propose a subsampling strategy. We conclude this section with a discussion of the algorithm.

\subsection{Details of Implementation}
We first notice that we can achieve a complete reconstruction iteratively only with a proper stopping criterion. Naturally, the recovered source should be capable of generating all $T$ measurements with errors below the noise level. The stopping criterion can be defined as follows:
\begin{align}
	\max_{1\le t\le T} \min_{\hat{a}_{t,j}\in\mathbb{C}} \|[\hat \mu_t]- Y_t\|_{2}< \sqrt{2K+1}\sigma,
\end{align}  
where $\hat{\mu_t} = \sum_{j=1}^n \hat{a}_{t,j}\delta_{\hat{y}_j}$.
We display the pseudo-algorithm for the IFF algorithm in \textbf{Algorithm \ref{algo_3}}.
\begin{algorithm}
 	\caption{IFF Algorithm}
 	\label{algo_3}
 	\DontPrintSemicolon
    \SetKwInOut{Input}{Input}
    \SetKwInOut{Output}{Output}
    \SetKwInOut{Init}{Initialization}
 	\Input{Sampling number: $K$, Original measurements: $\{Y_t\}_{t=1}^T$, noise level: $\sigma$.}
 	\Init{$S^{\left(0\right)}=\emptyset$, $\{Y_t^{\left(0\right)}\}=\{Y_t\}$, $r=0$.}

 		\While {$\gamma \ge \sqrt{2K+1}\sigma$}{
 		Proceed \textbf{Algorithm \ref{algo_2}} with input $\{Y_t^{\left(r\right)}\}$ and $S^{\left(r\right)}$, get the output $\{Y_t^{\left(r+1\right)}\}$ and $\{H_t^{\left(r\right)}\}$.\;
 
 		Proceed \textbf{Algorithm \ref{algo_1}} with input Hankel matrices $\{H_t^{\left(r+1\right)}\}$ and noise level $\sigma^{\left(r\right)}$, get the output $S'$. Let $S^{\left(r+1\right)} = S^{\left(r\right)}\cup S'$.\;
 		
 		$\gamma = \max_{1\le t\le T} \min_{\hat{a}_{t,j}} \|\mathcal F[\hat \mu_t]- Y_t\|_{2}$, where  $
 \text{supp}\ \hat{\mu}_t=S^{\left(r+1\right)}$.\;
 		}
 	\Return A set of reconstructed point sources: $S$.\;
\end{algorithm}
We point our that in \textbf{Algorithm \ref{algo_1}}, the input $\sigma^{\left(r\right)}$ should be adjusted in each iteration. This is mainly because of the change of SNR induced by annihilating filters, as discussed in Section \ref{sec:source removal}. Motivated by the estimation in Section \ref{subsec:snr_filter}, one choice of $\sigma^{\left(r\right)}$ is
$$
\sigma^{\left(r\right)} = C\left(\frac{\Omega}{Kd}\right)^{|S^{\left(r\right)}|}\sigma,
$$ 
where $C$ is an order-one constant and $d$ is a prior guess for the cluster distance either based on the reconstruction result from step 3 in \textbf{Algorithm \ref{algo_3}} or other prior information. We point out that the choice of $\sigma^{(r)}$ can be flexible in practice. 

\subsection{The subsampling strategy}
We now discuss the subsampling strategy mentioned in previous sections. We notice that the MUSIC algorithm in \textbf{Algorithm \ref{algo_1}} has a huge computational cost, especially when the size of the Hankel matrix is large. In the source localization step, however, we only need to handle nearly rank-one Hankel matrices due to the numerical focusing step. It motivates us to use a few sample points to form Hankel matrices. The subsampling strategy also helps preserve the SNR in \textbf{Algorithm \ref{algo_2}}, as described in Section \ref{sec:source removal}. To achieve subsampling for a given restricted row number of the Hankel matrix, $m$, we select the maximum sampling distance. It is worth noting that, owing to the source focusing step, one may even choose $m=2$ or $m=3$. See the numerical example in Section \ref{numerical results}. We summarize this subsampling approach in the  
pseudo-code of \textbf{Algorithm \ref{algo_2'}} Below.

\begin{algorithm}
	\caption{Source Filtering with Subsampling}
	\label{algo_2'}
	\DontPrintSemicolon
    \SetKwInOut{Input}{Input}
    \SetKwInOut{Output}{Output}
    \SetKwInOut{Init}{Initialization}
	\Input{Original measurements: $\{Y_t\}_{t=1}^T$, Recovered sources: $S =\{\hat{y}_p\}_{p=1}^P$, sampling distance $h = \frac{\Omega}{K}$, row number of Hankel matrix: $m$.}
        Calculate subsampling distance: $h' = \lceil \frac{|Y_t|}{2m-1+|S|} \rceil \cdot h$.\;
		Form the filter $F=\left(1,-e^{i\hat{y}_1 h'}\right)\ast\cdots\ast\left(1,-e^{i\hat{y}_P h'}\right)$.\;
		\For {$t\ =\ 1:T$}{
		Compute the discrete convolution $Y_t\ast F$\;
		Let $Y'_t=Y_t\ast F\left(P+1:2K+1\right)$.\;
        Compute the subsampling factor $s=\lceil \frac{|Y'_t|}{2m-1}\rceil$, the column number of Hankel matrix $l=[\frac{|Y'_t|}{s}]-m+1$.\;
        Form $m\times l$ Hankel matrix $H'_t$ using $Y_t'\left(1:s:|Y_t'|\right)$.\;
		}
	\Return $\{Y_t'\}_{t=1}^T$, $\{H'_t\}_{t=1}^T$.\;
\end{algorithm}

We note that in Steps 6 and 7 of \textbf{Algorithm \ref{algo_2'}}, the Hankel matrix is formed based on the principle that we minimize the number of uniformly distributed sample points while retaining as much high-frequency data as possible. By substituting \textbf{Algorithm \ref{algo_2}} with \textbf{Algorithm \ref{algo_2'}} in \textbf{Algorithm \ref{algo_3}}, the IFF Method with subsampling can be obtained.


\subsection{Discussion of IFF Algorithm}
In this subsection, we briefly discuss the IFF algorithm. First, we point out that from numerical experiments, the IFF algorithm is not sensitive to the choice of hyperparameters such as $\Gamma$ in \textbf{Algorithm \ref{algo_1}} and $\sigma^{(r)}$ in \textbf{Algorithm \ref{algo_3}}.
In the IFF algorithm, \textbf{Algorithm \ref{algo_1}} plays an essential role in the accuracy and efficiency of the whole method. It also assumes 
most of the computational power as source localization can be achieved using a subspace method applied to a small-sized Hankel matrix (even for $m=2$).
The acceleration of \textbf{Algorithm \ref{algo_1}} can be achieved in the following ways. First, we notice that optimization problems for different initial guesses are mutually independent and thus can be solved parallelly. Second, techniques to accelerate the optimization algorithm can also be applied to improve both efficiency and accuracy.

We also remark that the IFF algorithm presented here is only suitable for the case of point sources in a region, say of size comparable to several Rayleigh lengths, and the number of measurements is greater than the source number. For the general case of point sources in a large region and with multiscale structures, one can first decompose each measurement into "local measurements" that correspond to point sources in a small region or cluster, and then apply the IFF algorithm to each local measurement to reconstruct the point sources therein. The decomposition technique has been developed in \cite{liu2022measurement} and the 
scanning technique will be developed in a forthcoming paper.

In summary, IFF is an efficient algorithm that solves the local super-resolution problem with multiple measurements. We believe the idea can be generalized to other related problems. 

\section{Numerical Study}{\label{sec:numerical experiment}}
In this section, we conduct several groups of numerical experiments to show the phase transition phenomenon and the numerical behavior. For the IFF Method, in the experiments, we apply Nesterov Accelerated Gradient to solve the non-convex optimization problem proposed in (\ref{opt_pro}).
\subsection{Phase Transition Phenomenon}
We first introduce the super-resolution factor (SRF) which characterizes the ill-posedness of the super-resolution problem \cite{candes2014towards}. In the off-the-grid setting, SRF is defined by the ratio of Rayleigh limit and minimum separation distance. In our problem, the Rayleigh limit is $\frac{\pi}{\Omega}$ and the corresponding super-resolution factor is 
\begin{align}{\label{SRF}}
    SRF:= \frac{\pi}{\Omega \tau},
\end{align}
where $\tau=\min_{p\neq j}|y_p-y_j|$.

For the super-resolution problem with multiple measurements, we notice that the worst case is that all the illumination patterns are linearly dependent and the measurements are essentially one measurement. In such a case,  
the computational resolution is of the order $\mathcal{O}\left(\frac{\sigma}{m_{\min}}\right)^{\frac{1}{2n-1}}$, see \cite{9410626}. Here, $m_{\min} = \min_{j=1}|a_j|$.
The best case is that each measurement is associated with a single ``illuminated'' source. Following the discussions in Section \ref{subsec:mini_separation and error}, we predict that a phase transition phenomenon may occur. More precisely, we have demonstrated that a stable reconstruction is guaranteed if
\begin{align*}
    \log\left(SNR\right)>\left(2n-1\right)\log\left(SRF\right)+ \text{const}_1,
\end{align*}
and may fail if 
\begin{align*}
    \log\left(SNR\right)<n\log\left(SRF\right)+ \text{const}_2.
\end{align*}
Therefore, in the parameter space of $\log\left(SNR\right)-\log\left(SRF\right)$, there exist two lines with slopes $2n-1$ and $n$ respectively such that the cases above the former one are successfully reconstructed and the cases below the latter one are failed. As for the cases belonging to the intermediate region, the reconstruction result depends on the specifics of the illumination matrix. The phase transition is clearly illustrated in the numerical result below.

We set $\Omega=1$, and uniformly set 4 point sources in $[-1,1]$ with separation distance $\tau$, amplitude $1$, and the noise level $\sigma$. We set the illumination matrix,$L$, consisting of independent and identically distributed random variables. We perform 1000 times experiments with random choices of $d$, $L$, and $\sigma$ to reconstruct the point sources using \textbf{Algorithm \ref{algo_3}} and each time we use 10 measurements. Figure \ref{Fig:phase transition} shows the separation of the red points (successful reconstruction) and blue points (unsuccessful reconstruction) by two lines with slope $n$ and $2n-1$ respectively. The region in between is the phase transition region.
\begin{figure}[H]
\centering
\includegraphics[width=0.5\textwidth]{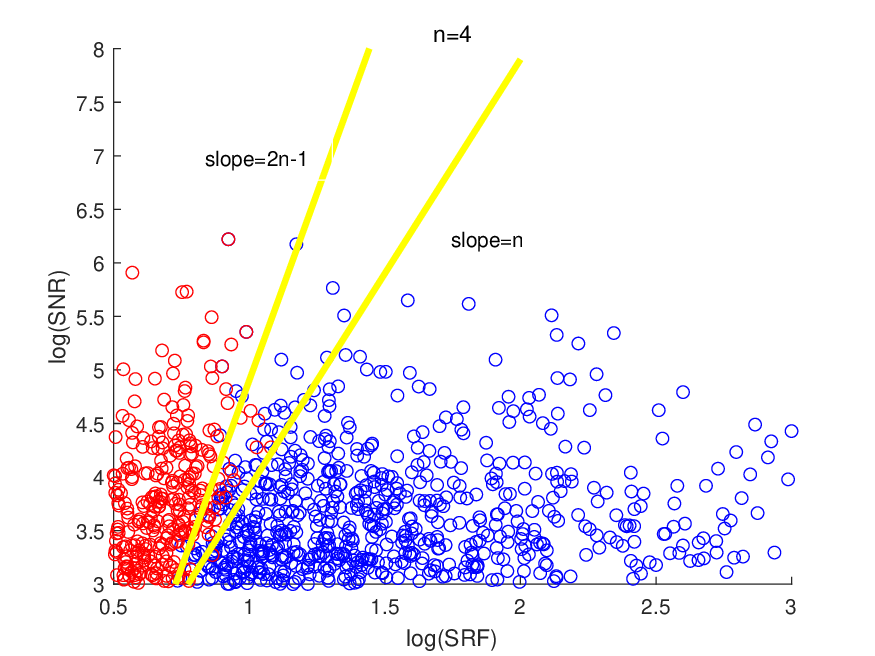}
\caption{Plot of successful and unsuccessful point source reconstruction by \textbf{Algorithm \ref{algo_3}} in the parameter space $\log\left(SNR\right)-\log\left(SRF\right)$. The red one represents the successful case and the blue one represents the unsuccessful case.}
\label{Fig:phase transition}
\end{figure}

\subsection{Reconstruction Behavior}
To demonstrate the numerical reconstruction ability, we first try the IFF Method on a challenging regime. We set $n=4$, $\Omega=1$, $\sigma=10^{-4}$ and 
\begin{align*}
    \mu = \delta_{-0.75}+\delta_{-0.25}+\delta_{0.25}+\delta_{0.75},
\end{align*}
For the illumination matrix, $L$, we set each element to be i.i.d. random variable that follows $\mathcal{U}[1,1+\sqrt{3}]$. Under this setup, we test the traditional subspace method using a single measurement. The result shows that the reconstructed source number is three, and the reconstructed source position is therefore inaccurate. We perform 1000 random experiments using the IFF Method (the randomness comes from the illumination and noise from the measurement). Each time, we use 10 measurements to reconstruct the point source's position. The mean of position from 1000 experiments is $\left(\hat{y}_1,\hat{y}_2,\hat{y}_3,\hat{y}_4\right)=\left(-0.7486,-0.2486,0.2481,0.7490\right)$. To compare with the standard method in multi-illumination LSE that aligns all the measurements and then applies the subspace method. We apply the standard method to the same data. It should be pointed out that throughout the paper when we apply the the Aligned MUSIC Method, we provide the prior information of the source number.
The numerical result shows the mean of position is $\left(\hat{y}_1,\hat{y}_2,\hat{y}_3,\hat{y}_4\right)=\left(-0.7497,-0.2492,0.2493,0.7496\right)$. The variance of the two methods is both in the order $\mathcal{O}\left(10^{-4}\right)$. Considering the noise level is $10^{-4}$, we claim the numerical behavior of \textbf{Algorithm \ref{algo_3}} is comparable to the standard method.
\subsection{Running Time of the IFF Method}
To test the running time of the IFF Method, we set $\Omega = 1$, $n=3$, $SNR\sim \mathcal{O}\left(10^{2}\right)$, and
\[
\mu = \delta_{-0.9}+\delta_{0}+\delta_{0.9}.
\]
In the experiment of the IFF Method, we fix the Hankel matrix size to be 2-by-2 and we do not perform the parallel computing of \textbf{Algorithm \ref{algo_1}}. In each setup, we conduct 10 random experiments, Figure \ref{Fig:number-time} shows the averaged running time of two methods. We see the running time of the two methods is comparable. We point out that the IFF Method can be further accelerated by parallel computing and better choices of step size in solving the optimization problem. Meanwhile, we can replace MUSIC with MP method, which is more computationally efficient, in \textbf{Algorithm \ref{algo_1}}. We expect the IFF Method to have a greater advantage in terms of speed in two-dimensional super-resolution problems, which will be our future work.
\begin{figure}[H]
\centering
\includegraphics[width=0.5\textwidth]{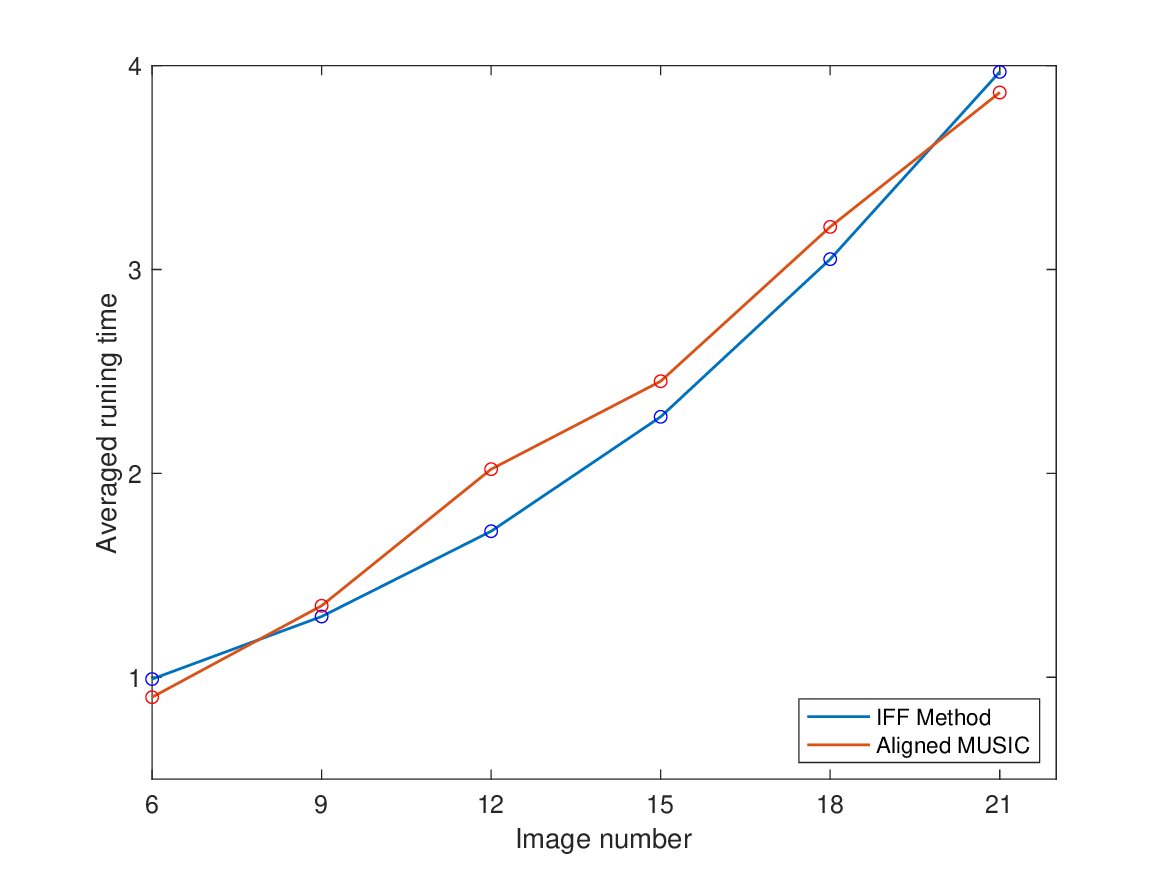}
\caption{Plot of averaged running time of the IFF Method and traditional aligned MUSIC method with increasing image number.}
\label{Fig:number-time}
\end{figure}

Next, we test the IFF Method with the application of parallel computing for the optimization problem with different initial guesses. We apply parallel computing for the IFF Method using 8 cores. Three groups of experiments are conducted to compare the IFF Method and the Aligned MUSIC Method in terms of efficiency and performance.

We first set $\Omega = 1$, $n=3$, $T=10$, $SNR\sim \mathcal{O}\left(10^{2}\right)$, and
\[
\mu = \delta_{-1}+\delta_{0}+\delta_{1}.
\]
We fix the Hankel matrix size to be 2-by-2. We conduct 50 random experiments for the Aligned MUSIC Method and the IFF Method respectively. 
Denote $y$ as the ground truth and $\hat{y}$ as the reconstruction result, we define the $\ell_2$ relative error to be
\begin{align}
    e = \frac{\|\hat{y}-y\|_2}{\|y\|_2}.
\end{align}
The comparison of averaged running time and reconstruction error is shown in Figure \ref{Fig:three_compare}.
\begin{figure}[ht]
	\centering
    \subfloat[Averaged running time comparison.]{
	\includegraphics[width=0.4\textwidth]{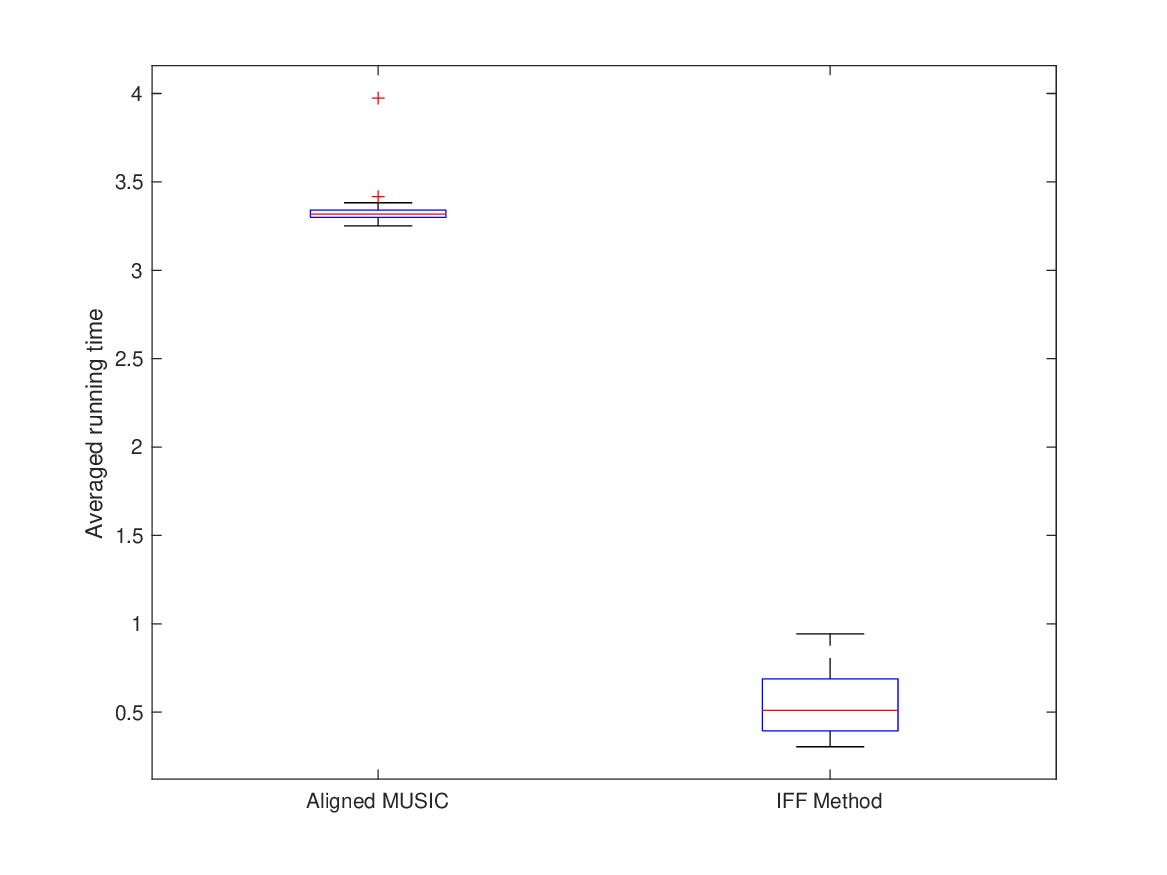}}
    \quad
    \subfloat[Reconstruction error comparison.]{
	\includegraphics[width=0.4\textwidth]{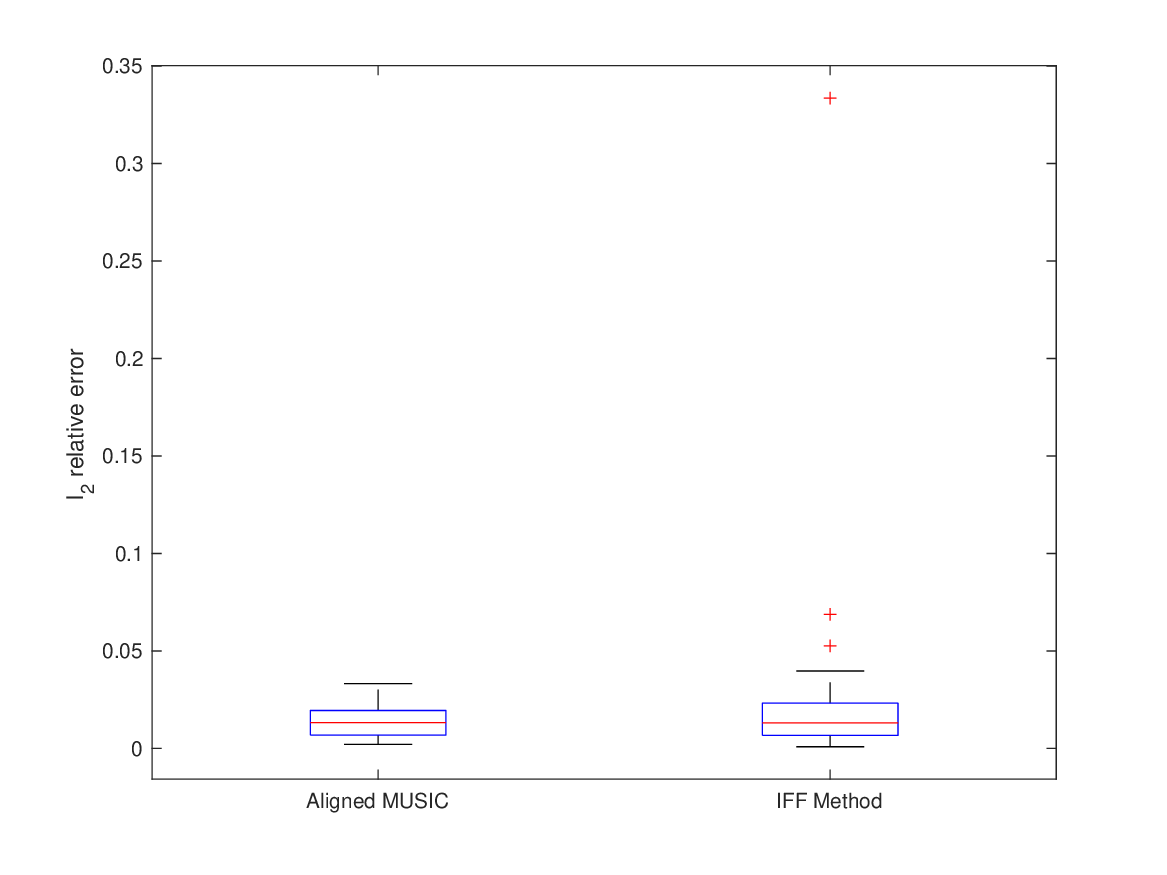}}
 \caption{Numerical result for reconstructing three point sources using the Aligned MUSIC Method and the IFF Method.}
\label{Fig:three_compare}
\end{figure}
We observe that the IFF Method achieves comparable accuracy and we notice that with the application of parallel computing, the IFF Method shows the advantage on the computational efficiency.\\
For the case where four point sources are closely positioned, we set $\Omega = 1$, $n=4$, $T=10$, $SNR\sim \mathcal{O}\left(10^{2}\right)$, and
\[
\mu = \delta_{-3}+\delta_{-1.5}+\delta_{0}+\delta_{1.5}.
\]
We fix the row number of the Hankel matrix to be 3 and apply parallel computing. We conduct 50 random experiments for the Aligned MUSIC Method and the IFF Method respectively. The comparison of averaged running time and reconstruction error is shown in Figure \ref{Fig:four_compare}.
\begin{figure}[ht]
	\centering
    \subfloat[Averaged running time comparison.]{
	\includegraphics[width=0.4\textwidth]{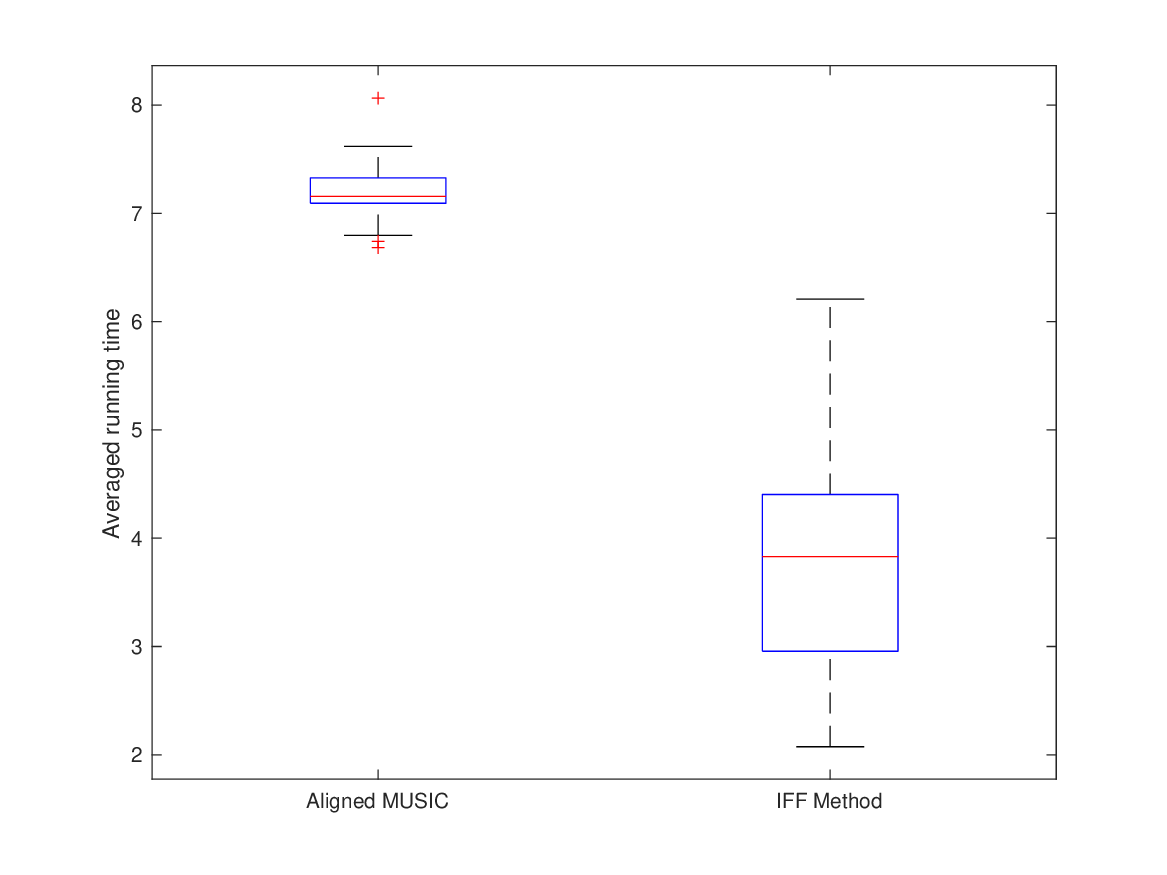}}
    \quad
    \subfloat[Reconstruction error comparison.]{
	\includegraphics[width=0.4\textwidth]{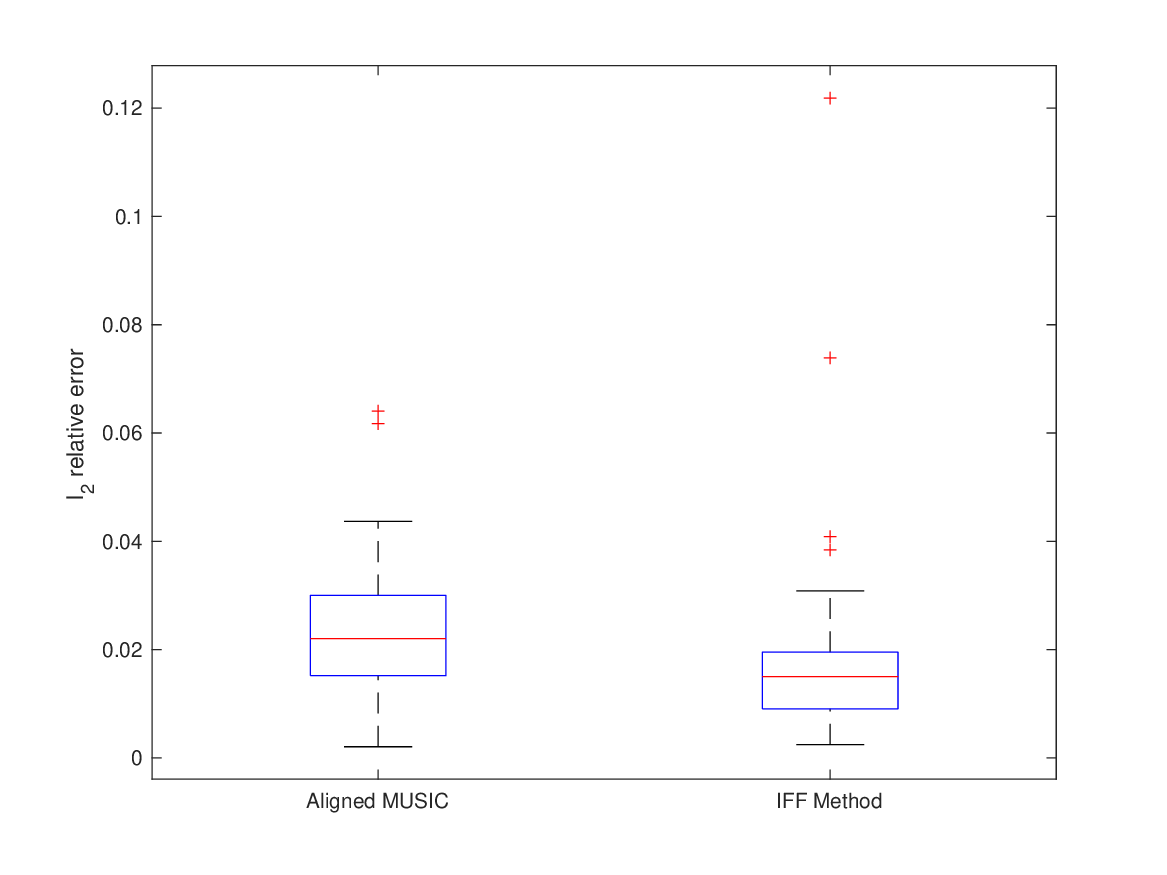}}
 \caption{Numerical result for reconstructing four point sources using the Aligned MUSIC Method and the IFF Method.}
\label{Fig:four_compare}
\end{figure}
We notice that for the setup above, the IFF Method achieves accuracy at least comparable to the Aligned MUSIC Method. The efficiency of the IFF Method can also be observed.\\
We also test the case where a larger number of sources are separated with distance $\pi$. We set $\Omega = 1$, $n=7$, $T=15$, $SNR\sim \mathcal{O}\left(10^{2}\right)$, and
\[
\mu = \sum_{j=-3}^3 \delta_{j\pi}.
\]
We fix the row number of the Hankel matrix to be 3 and apply parallel computing. We conduct 50 random experiments for the Aligned MUSIC Method and the IFF Method respectively. The comparison of averaged running time and reconstruction error is shown in Figure \ref{Fig:seven_compare}.
\begin{figure}[ht]
	\centering
    \subfloat[Averaged running time comparison.]{
	\includegraphics[width=0.4\textwidth]{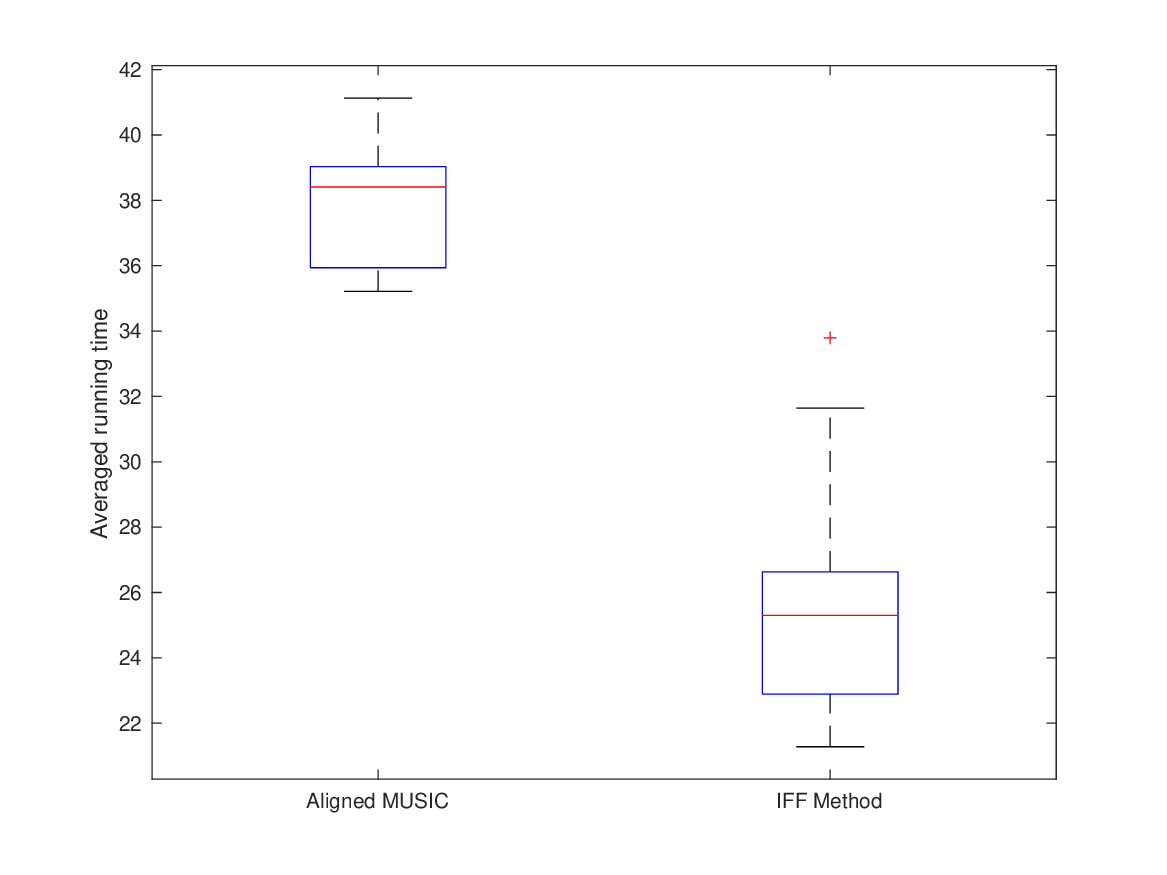}}
    \quad
    \subfloat[Reconstruction error comparison.]{
	\includegraphics[width=0.4\textwidth]{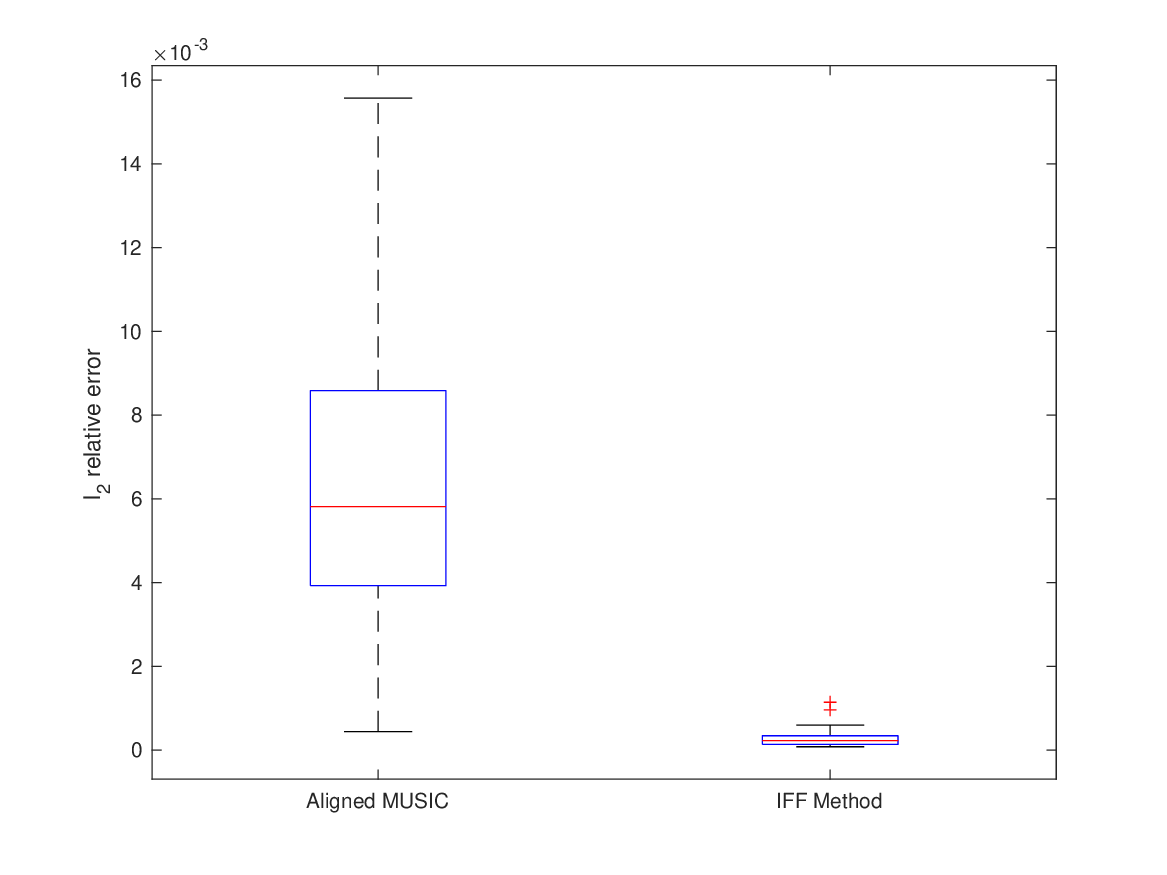}}
 \caption{Numerical result for reconstructing seven point sources using the Aligned MUSIC Method and the IFF Method.}
\label{Fig:seven_compare}
\end{figure}
We observe that the IFF Method achieves stable reconstruction in a shorter running time.

\subsection{Other Numerical Results}{\label{numerical results}}
In this subsection, we present three groups of experiments to test the robustness of the IFF algorithm to the image number $T$, the signal-to-noise ratio $SNR$, and the Hankel matrix size $m$ used in \textbf{Algorithm \ref{algo_2'}} respectively. All experiments in this subsection are conducted under $\Omega = 1$, $n=3$. In each group, we conduct 20 independent random experiments.

First, we set $SNR\sim \mathcal{O}\left(10^{2}\right)$, $m=2$, $T=6$ and
\[
\mu = \delta_{-d}+\delta_{0}+\delta_{d},
\]
where the separation distance, $d$, is chosen to be near the theoretical limit as in (\ref{separation_dist_upper}). We say the IFF Method succeeds if it finds three different sources from the multiple measurements. Figure \ref{Fig:number-successful_rate} shows the success rate under two different choices, $d=0.8$ and $d=0.9$. The results demonstrate that IFF can successfully reconstruct all the point sources provided that there are enough linearly independent measurements. Moreover, the reconstruction is easier when the point sources are more separated. 

\begin{figure}[H]
\centering
\includegraphics[width=0.5\textwidth]{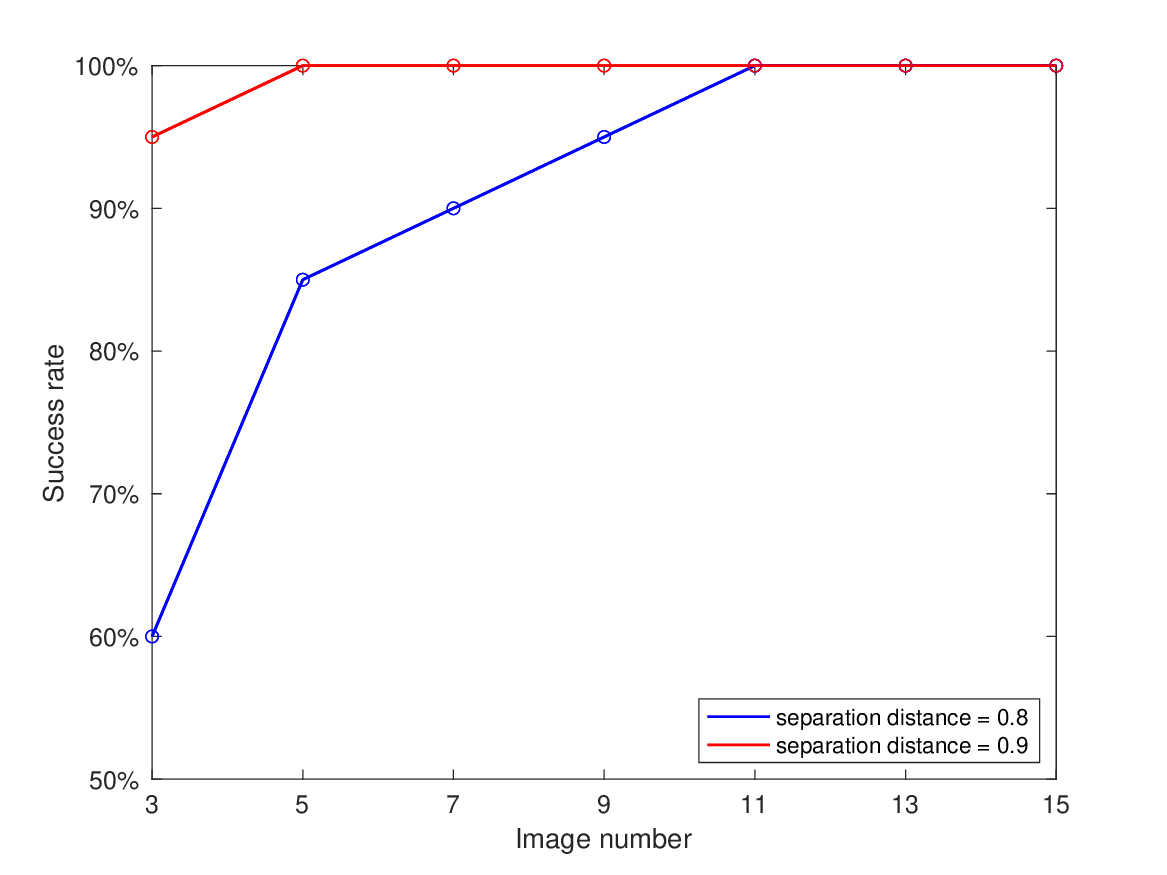}
\caption{Plot of success rate with increasing image numbers under two different separation conditions.}
\label{Fig:number-successful_rate}
\end{figure}

Second, we set $m=2$, $T=6$ and 
\[
\mu = \delta_{-0.9}+\delta_{0}+\delta_{0.9}.
\]
Figure \ref{Fig:noise-error} shows the $\ell_2$ relative error under different noise levels. The results clearly show the robustness of the IFF algorithm to the noises.

\begin{figure}[H]
\centering
\includegraphics[width=0.5\textwidth]{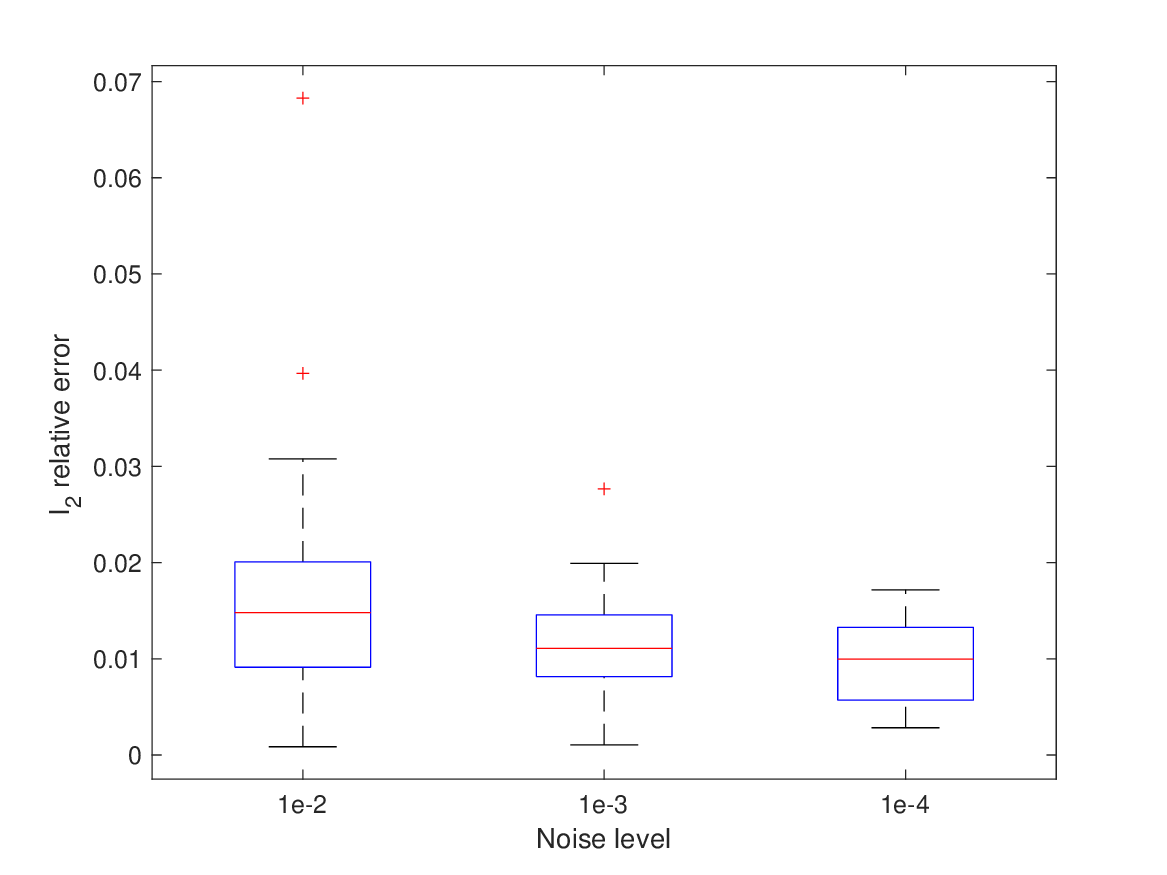}
\caption{Boxplot of $\ell_2$ relative error with different noise level.}
\label{Fig:noise-error}
\end{figure}
Finally, we study how the matrix size $m$ affects the reconstruction error through numerical experiments. We set $T=6$, $SNR\sim \mathcal{O}\left(10^{2}\right)$ and 
\[
\mu = \delta_{-0.9}+\delta_{0}+\delta_{0.9}.
\]
We observe that provided the noise level is $\mathcal{O}\left(10^{-2}\right)$, the result given by $2$-by-$2$ Hankel matrices is already satisfactory. We note that a larger choice of matrix size may bring better stability in the localization, but it also introduces higher computational costs.
\begin{figure}[H]
\centering
\includegraphics[width=0.5\textwidth]{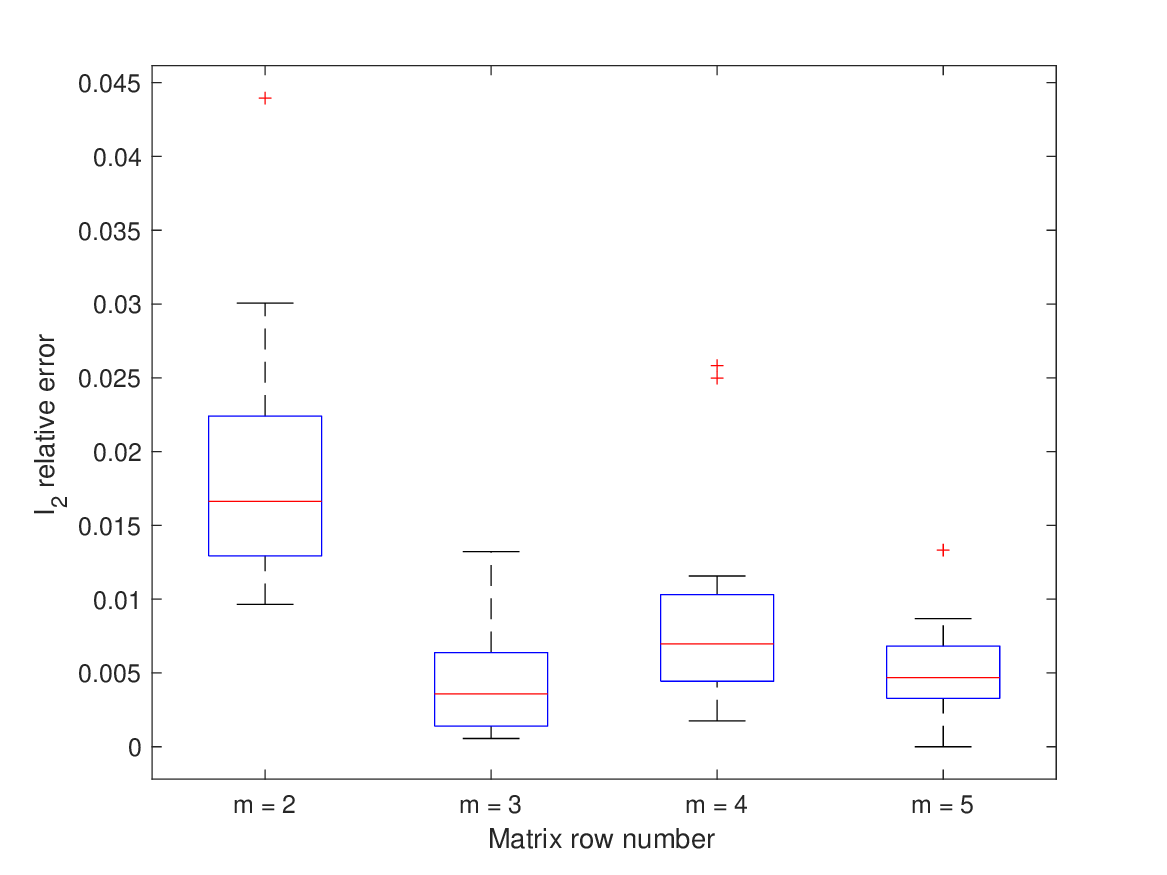}
\caption{Boxplot of $\ell_2$ relative error with different matrix size}
\label{Fig:size-error}
\end{figure}

\section{Conclusion}{\label{sec:conclusion}}
In this paper, we propose the IFF algorithm for the multiple measurements super-resolution problem in 1D and provided theoretical analysis for the method behind it. We use several groups of numerical experiments to illustrate the phase transition phenomenon and its efficiency and robustness. With comparable numerical results to standard methods, the advantages of the IFF algorithm can be summarized in the following three aspects. First, unlike standard subspace-method-based algorithms, the proposed one does not need prior information on the source number. Second, the subsampling strategy allows the reconstruction based on small-size Hankel matrices. Third, the algorithm can be easily paralleled. The proposed algorithm may offer a new perspective in designing algorithms for super-resolution and other related problems. We particularly expect this algorithm/method to play a greater role in problems where the singular value decomposition of large matrices by traditional subspace methods requires high computational power.

\section{Proofs of Results in Section \ref{sec:source focusing}}{\label{sec:proofs}}
\subsection{Proof of Proposition \ref{contrast}}
First, we denote
\begin{align}{\label{phi_omega}}
    \phi_s\left(\omega\right)=\left(1,\omega,\omega^2,\cdots,\omega^s\right).
\end{align}
We define $\mathcal{H}_j=a_je^{-iy_j\Omega}\phi_K(e^{iy_j\frac{\Omega}{K}})\phi^{\top}_K(e^{iy_j\frac{\Omega}{K}})$ and $\hat{\mathcal{H}}_j=a_j\phi_K(e^{iy_j\frac{\Omega}{K}})\phi^*_K(e^{iy_j\frac{\Omega}{K}})$. It is easy to check that $\|\mathcal{H}_j\|_2=\|\mathcal{H}_j\|_F=\|\hat{\mathcal{H}}_j\|_F$.\\
Let $H$ be the Hankel matrix generated by $\mu$ as in (\ref{hankel}). Clearly, $H$ has the following decomposition:
\begin{align}
    H = \sum_{j=1}^n \mathcal{H}_j + \left(\begin{array}{ccc}
W_{-K} & \cdots & W_{0} \\
\vdots & & \vdots \\
W_{0} & \cdots & W_{K}
\end{array}\right) \triangleq \sum_{j=1}^n \mathcal{H}_j+\Delta.
\end{align}
A straight-forward calculation gives
\begin{align}
\|\hat{\mathcal{H}}_j\|_F^2=|a_j|^2Tr\left(\hat{\mathcal{H}^*_j}\hat{\mathcal{H}_j}\right)=|a_j|^2|\phi_K(e^{iy_j\frac{\Omega}{K}})|^4=\left(K+1\right)^2|a_j|^2.
\end{align}
Notice that $H = \mathcal{H}_s+\sum_{t\ne s}\mathcal{H}_t+\Delta$. By Weyl's theorem, we have
\begin{align}
    |\|H\|_2-\|\mathcal{H}_s\|_2|
    &\le \|\sum_{t\ne s}\mathcal{H}_t+\Delta\|_2
    \le \|\sum_{t\ne s}\mathcal{H}_t+\Delta\|_F\notag\\
    &\le \sum_{t\ne s}\|\mathcal{H}_t\|_F+\|\Delta\|_F
    \le \left(K+1\right)\left(\sum_{t\ne s} |a_t|+\sigma\right).
\end{align}
Therefore
\begin{align}
    \|H\|_2\ge \|\mathcal{H}_s\|_2-\left(K+1\right)\left(\sum_{t\ne s} |a_t|+\sigma\right)
    = \left(K+1\right)\left(|a_s|-\sum_{t\ne s} |a_t|-\sigma\right).
\end{align}
On the other hand,
\begin{align}
    \|H\|_F\le\sum_{j=1}^n\|\mathcal{H}_j\|_F+\|\Delta\|_F
    = \left(K+1\right)\left(|a_s|+\sum_{t\ne s} |a_t|+\sigma\right).
\end{align}
Combining the above two inequalities, we have
\begin{align}
    \frac{Tr\left(N\right)^2}{Tr\left(N^*N\right)}
    \le \left(\frac{\|H\|_F}{\|H\|_2}\right)^4
    \le \left(\frac{|a_s|+\sum_{t\ne s} |a_t|+\sigma}{|a_s|-\sum_{t\ne s} |a_t|-\sigma}\right)^4.
\end{align}

\subsection{Proof of Proposition \ref{threshold}}
Similar to the argument in the proof of Proposition \ref{contrast}, we write 
\begin{align}
H & = M\left(\begin{array}{ccc}
e^{iy\omega_{-K}} & \cdots & e^{iy\omega_{0}} \\
\vdots & & \\
e^{iy\omega_0} & \cdots & e^{iy\omega_{K}}
\end{array}\right)+\left(\begin{array}{ccc}
W_{-K} & \cdots & W_{0} \\
\vdots & & \vdots \\
W_{0} & \cdots & W_{K}
\end{array}\right) \\
  & \triangleq X+\Delta.
\end{align}

\textbf{Step 1.}
We list the singular values of $H$ in a descending order as:
\[
\hat{\sigma}_1\ge\hat{\sigma}_2\ge \cdots\ge\hat{\sigma}_{K+1}.
\]
Similarly, we list the singular values of $X$:
\[
\sigma_1\ge\sigma_2\ge \cdots\ge\sigma_{K+1}.
\]
A straight-forward calculation shows that
\begin{align}{\label{sigma_cal}}
    \sigma_1 = M\left(K+1\right),\ \ \sigma_j = 0,\ j\ge 2.
\end{align}
By Weyl's theorem, for $1\le j \le K+1$, we have 
\begin{align}{\label{Wely's thm}}
    |\hat{\sigma}_j-\sigma_j|\le \|\Delta\|_2 \le\|\Delta\|_F\le \left(K+1\right)\sigma. 
\end{align}

\textbf{Step 2.}
Define $r_j = \frac{\hat{\sigma}_j}{\hat{\sigma}_1}$ for $1\le j\le K+1$.
Combining (\ref{sigma_cal}) and (\ref{Wely's thm}), we have 
\begin{align}
    r_j \le \frac{\sigma}{M-\sigma}\le \frac{2\sigma}{M}, \ 2\le j\le K+1.
\end{align}
By the fact that the function $g\left(t\right):=\frac{\left(1+t\right)^2}{1+t^2}$ is increasing on interval $\left(0,1\right)$, we have
\begin{align}
\frac{Tr\left(N\right)^2}{Tr\left(N^*N\right)} 
=\frac{\left(\sum_{j=1}^{K+1} \hat{\sigma}^2_j\right)^2}{\sum_{j=1}^{K+1} \hat{\sigma}^4_j}
 =\frac{\left(\sum_{j=1}^{K+1} r^2_j\right)^2}{\sum_{j=1}^{K+1} r^4_j}
 \le& \left(1+4K\frac{\sigma^2}{M^2}\right)^2.
\end{align}

\subsection{Proof of Proposition \ref{L_eff u not 0}}
Denote $1_T=\left(1,1,\cdots,1\right)^{\top}\in \mathbb{C}^T$ and 
let $Q_{j} = I_{T\times T}-L_{j}\left(L_{j}^* L_{j}\right)^{-1}L_{j}^*$. \\
It is easy to see that for each $1\le j\le T$, $Q_{j}$ is a positive semi-definite projection matrix with $Tr\left(Q_{j}\right) = T-n+1$, $\|Q_{j}\|_2 = 1$ and $\|Q_{j}\|_F^2 = Tr\left(Q_{j}^*Q_{j}\right)= Tr\left(Q_{j}\right)=T-n+1$.\\
Therefore for each $j = 1,2,\cdots,T$, we have
\begin{align}{\label{L eff estimation}}
	\mathbb{E} \|\left(I-\mathcal{P}_{j}\right)\alpha_j\|_2^2 
	&= \mathbb{E}\left(\alpha_j^*{Q_{j}^*}{Q_{j}}\alpha_j\right) \notag
	= \mathbb{E}\left(\mathbb{E}\left(\alpha_j^*{Q_{j}^*}Q_{j}\alpha_j|Q_{j}\right)\right) \notag\\
	&= \mathbb{E}\left(\mathbb{E}\left(Tr\left(\alpha_j^*{Q_{j}^*}Q_{j}\alpha_j\right)|Q_{j}\right)\right) \notag\\
	&= v^2 \mathbb{E}\left(Tr\left(Q_{j}\right)u^2+ 1_T^*Q_{j}1_T\right)\notag\\
	&\ge \left(T-n+1\right)v^2.
\end{align}

\subsection{Proof of Proposition \ref{L_eff}}
From the argument of the proof of Proposition \ref{L_eff u not 0}, when $u=0$, we have 
\begin{align}
\mathbb{E}\  R_{nn}^2 = \left(T-n+1\right)v^2.
\end{align}
Combining the equality above and the Hanson-Wright inequality \cite{rudelson2013hanson}, we obtain the desired result.

\subsection{Proof of Theorem \ref{separation_dist_upper}}
We only give the proof for $n$ is even, the case when $n$ is odd can be proved in a similar way.
\textbf{Step 1.} Let 
\begin{align}{\label{division with residue}}
    2K+1=\left(n+1\right)r+q,
\end{align}
where $r,\ q$ are nonnegative integers with $r\ge1$ and $0\le q<n+1$. We denote $\theta_j=y_j\frac{r\Omega}{K}$, $j=1,\cdots,n$. For $y_j\in[-\frac{\pi}{2\Omega},\frac{\pi}{2\Omega}]$, in view of (\ref{division with residue}), it is clear that
\begin{align}
    \theta_j=y_j\frac{r\Omega}{K}\in[-\frac{\pi}{2},\frac{\pi}{2}], \ \ \ j=1,\cdots,n.
\end{align}
We write $\tau' = \tau\frac{r\Omega}{K}$ to be the minimum separation distance between $\{\theta_j\}$, define $\mathcal{S} = \{y\in[-\frac{\pi}{2\Omega},\frac{\pi}{2\Omega}], \min_{j}|y-y_j|\ge \frac{\tau}{2}\}$, $\mathcal{S'} = \{\theta\in[-\frac{\pi}{2},\frac{\pi}{2}], \min_{j}|\theta-\theta_j|\ge \frac{\tau'}{2}\}$.\\
\textbf{Step 2.}
For $\hat{\mu} = \sum_{j=1}^n \hat{a}_j\delta_{y_j}$, we notice that
\begin{align}
    [\hat{\mu}]-[\mu]=\left(\mathcal{F}\hat{\mu}\left(\omega_{-K}\right),\mathcal{F}\hat{\mu}\left(\omega_{-K+1}\right),\cdots,\mathcal{F}\hat{\mu}\left(\omega_K\right)\right)^{\top}-\left(\mathcal{F}\mu\left(\omega_{-K}\right),\mathcal{F}\mu\left(\omega_{-K+1}\right),\cdots,\mathcal{F}\mu\left(\omega_K\right)\right)^{\top},
\end{align}
where $\omega_{-K}=-\Omega$, $\omega_{-K+1}=-\Omega+h$, $\cdots$, $\omega_K=\Omega$ and $h=\frac{\Omega}{K}$. In what follows, We denote the partial measurement as $\omega_{r,t}=-\Omega+trh$, $0\le t\le n$. We then have
\begin{align}
    \left(\mathcal{F}\hat{\mu}\left(\omega_{r,0}\right),\mathcal{F}\hat{\mu}\left(\omega_{r,1}\right),\cdots,\mathcal{F}\hat{\mu}\left(\omega_{r,n}\right)\right)^{\top}-\left(\mathcal{F}\mu\left(\omega_{r,0}\right),\mathcal{F}\mu\left(\omega_{r,1}\right),\cdots,\mathcal{F}\mu\left(\omega_{r,n}\right)\right)^{\top} = B_1\hat{a}-M\tilde{B}_1,
\end{align}
where $\hat{a}=\left(\hat{a_1},\hat{a_2},\cdots,\hat{a_n}\right)^{\top}$, $\tilde{B}_1=\left(e^{iy\omega_{r,0}},e^{iy\omega_{r,1}},\cdots,e^{iy\omega_{r,n}}\right)$ and
\begin{align}
    B_1=\left(\begin{array}{ccc}
e^{iy_1\omega_{r,0}} & \cdots & e^{iy_n\omega_{r,0}} \\
e^{iy_1\omega_{r,1}} & \cdots & e^{iy_n\omega_{r,1}} \\
\vdots & \vdots &\vdots \\
e^{iy_1\omega_{r,n}} & \cdots & e^{iy_n\omega_{r,n}}
\end{array}\right)
\end{align}
It is clear that 
\begin{align}{\label{infty norm to 2 norm}}
    \min_{\hat{a}\in\mathbb{C}^n,y\in\mathcal{S}}\|[\hat{\mu}]-[\mu]\|_{\infty}
    \ge\min_{\hat{a}\in\mathbb{C}^n,y\in\mathcal{S}}\|B_1\hat{a}-M\tilde{B}_1\|_{\infty}
    \ge\frac{1}{\sqrt{n+1}}\min_{\hat{a}\in\mathbb{C}^n,y\in\mathcal{S}}\|B_1\hat{a}-M\tilde{B}_1\|_{2}.
\end{align}
\textbf{Step 3.}
Let $\theta=y\frac{r\Omega}{K}\in[-\frac{\pi}{2},\frac{\pi}{2}]$ and $\phi_s\left(\omega\right)=\left(1,\omega,\omega^2,\cdots,\omega^s\right)$ defined as in (\ref{phi_omega}).\\
Notice that 
\begin{align}
    B_1 = \left(\phi_n(e^{i\theta_1}),\phi_n(e^{i\theta_2}),\cdots,\phi_n(e^{i\theta_n})\right)\text{diag}\left(e^{-iy_1\Omega},e^{-iy_2\Omega},\cdots,e^{-iy_n\Omega}\right).
\end{align}
We have 
\begin{align}\label{2norm rewrite}
    \min_{\hat{a}\in\mathbb{C}^n,y\in\mathcal{S}}\|B_1\hat{a}-M\tilde{B}_1\|_{2}
    =
    \min_{\hat{a}\in\mathbb{C}^n,y\in\mathcal{S}}\|D_1\hat{a}-\tilde{M}\tilde{D}_1\|_{2},
\end{align}
where $D_1=\left(\phi_n\left(e^{i\theta_1}\right),\phi_n\left(e^{i\theta_2}\right),\cdots,\phi_n\left(e^{i\theta_n}\right)\right)$, $\tilde{D}_1=\phi_n\left(e^{i\theta}\right)$, $\tilde{M}=e^{-iy\Omega}M$.\\
\textbf{Step 4.} By Theorem 3.1 in \cite{9410626}, we have
\begin{align}\label{estimation on 2norm}
    &\min_{\hat{a}\in\mathbb{C}^n,y\in\mathcal{S}}\|D_1\hat{a}-\tilde{M}\tilde{D}_1\|_{2}
     \ge \min_{\theta\in\mathcal{S'}}\ \frac{M}{2^n}|\prod_{j=1}^n \left(e^{i\theta}-e^{i\theta_j}\right)| \notag \\
    & \ge \min_{\theta\in\mathcal{S'}}\ \frac{M}{\pi^n}\prod_{j=1}^n |\theta-\theta_j|
    \quad\left(\text{by $|e^{i\theta}-e^{i\theta_j}|\ge\frac{2}{\pi}|\theta-\theta_j|$, for all $\theta,\theta_j\in[-\frac{\pi}{2},\frac{\pi}{2}]$}\right)\notag\\
    & =\frac{M}{\pi^n} \left(\frac{\tau'}{2}\right)^n\left(n-1\right)!
    \quad\left(\text{by Lemma \ref{lem_2}}\right)\notag\\
    & \ge M\left(\frac{\Omega}{\pi}\right)^n \frac{\left(n-1\right)!}{\left(n+1\right)^n}\frac{\tau^n}{2^n}
    \quad\left(\text{by the fact that $\frac{r\Omega}{K}\ge\frac{\Omega}{n+1}$}\right)\notag\\
    & \ge M\left(\frac{\Omega}{\pi}\right)^n \frac{\tau^n}{2^n} \frac{1}{e^n}\left(\frac{n-1}{n+1}\right)^n\sqrt{\frac{2\pi e^2}{n-1}}.
    \quad\left(\text{by Stirling's formula}\right)
\end{align}
\textbf{Step 5.} Combining (\ref{infty norm to 2 norm}), (\ref{2norm rewrite}) and (\ref{estimation on 2norm}), we have
\begin{align}
\min_{\hat{a}\in\mathbb{C}^n,y\in\mathcal{S}}\|[\hat{\mu}-\mu]\|_{\infty}
    \ge M\left(\frac{\Omega}{\pi}\right)^n \frac{\tau^n}{2^n} \frac{1}{e^n}\left(\frac{n-1}{n+1}\right)^n\sqrt{\frac{2\pi e^2}{n^2-1}}.
\end{align}
By the separation distance condition in (\ref{separation distance}) and the inequality that $\frac{2\left(n+1\right)}{n-1}\left(\frac{n^2-1}{2\pi e}\right)^{\frac{1}{2n}}<3.03$, we have 
\begin{align}
   \min_{\hat{a}\in\mathbb{C}^n,y\in\mathcal{S}}\|[\hat{\mu}]-[\mu]\|_{\infty}
    \ge \sigma'. 
\end{align}
It then follows that if $\{\delta_{y_j}\}_{j=1}^n$ is $\sigma'$-admissible to $\mu$ then $\mu$ need to be within the $\frac{\tau}{2}$-neighborhood of $\{\delta_{y_j}\}_{j=1}^n$.

\subsection{Proof of Proposition \ref{separation_dist_lower}}
We only give the proof for $n$ is even, the case when $n$ is odd can be proved in a similar way. We write $n=2p$.

\textbf{Step 1.} Let 
\begin{align}
	\tau = \dfrac{0.96e^{-\frac{3}{2}}}{\Omega}\left(\frac{\sigma'}{M}\right)^{\frac{1}{n}}
\end{align}
and $y_1=-p\tau$, $y_2=-\left(p-1\right)\tau$, $\cdots$, $y_p=-\tau$, $y_{p+1}=0$, $\cdots$, $y_{2p+1}=p\tau$. Consider the following linear system
\begin{align}{\label{lin_sys}}
	Aa=0,
\end{align}
where $A=\left(\phi_{2p-1}\left(t_1\right),\cdots,\phi_{2p-1}\left(t_{2p+1}\right)\right)$. Since this linear system is underdetermined and $A$ has full column rank, there exists a solution whose elements are all nonzero real numbers. By a scaling of $a$, we can assume that $|a_k| = \min_{1\le j \le 2p+1}|a_j|=M$.\\
We define $\hat \mu=\sum_{j\ne k}a_j\delta_{y_j}$, $\mu=-a_{k}\delta_{y_{k}}$ and show $\|[\hat{\mu}]-[\mu]\|_{\infty}\le \sigma'$ next.

\textbf{Step 2.} Observe that
\begin{align}
\|[\mu] - [\mu]\|_{\infty} \le \max _{\omega \in[-\Omega,\Omega]}|\mathcal{F}\gamma\left(\omega\right)|,
\end{align}
where $\gamma = \sum_{j=1}^{2p+1}a_j\delta_{y_j}$ and 
\begin{equation}{\label{multi_exp}}
\mathcal{F}\gamma\left(\omega\right)=\sum_{j=1}^{2 p+1} a_{j} e^{i y_j \omega}=\sum_{j=1}^{2 p+1} a_{j} \sum_{k=0}^{\infty} \frac{\left(i y_j \omega\right)^{k}}{k !}=\sum_{k=0}^{\infty} Q_{k}\left(\gamma\right) \frac{\left(i \omega\right)^{k}}{k !},
\end{equation}
where $Q_{k}\left(\gamma\right)=\sum_{j=1}^{2 p+1}a_jy_j^k$. By (\ref{lin_sys}), we have $Q_{k}\left(\gamma\right)=0$ for $0\le k\le2p-1$. We then estimate $Q_{k}\left(\gamma\right)$ for $k\ge 2p$.

\textbf{Step 3.} We reorder $a_j$'s in the following way
\begin{align}
	M = |a_{j_1}|\le|a_{j_2}|\le\cdots\le|a_{j_{2p+1}}|.
\end{align}
From (\ref{lin_sys}), we derive that
\begin{align}
a_{j_{1}} \phi_{2p-1}\left(y_{j_{1}}\right)=\left(\phi_{2p-1}\left(y_{j_{2}}\right), \cdots, \phi_{2p-1}\left(y_{j_{2p+1}}\right)\right)\left(-a_{j_{2}}, \cdots,-a_{j_{2p+1}}\right)^{\top},
\end{align}
and hence
\begin{align}
a_{j_{1}}\left(\phi_{2p-1}\left(y_{j_{2}}\right), \cdots, \phi_{2p-1}\left(y_{j_{2p+1}}\right)\right)^{-1} \phi_{2p-1}\left(y_{j_{1}}\right)=\left(-a_{j_{2}}, \cdots,-a_{j_{2p+1}}\right)^{\top}.
\end{align}
By Lemma $3.5$ in \cite{9410626}, we have
\begin{align}
a_{j_{2p+1}}=- a_{j_{1}} \prod_{2 \leq q \leq 2p} \frac{y_{j_{1}}-y_{j_{q}}}{y_{j_{2p+1}}-y_{j_{q}}}.
\end{align}
Therefore
\begin{align}
|a_{j_{2p+1}}| \le \frac{\left(2p\right) !}{p!\cdot p!}|a_{j_{1}}| \le 2^{2p} M,
\end{align}
and consequently
\begin{align}
\sum_{j=1}^{2p+1}\left|a_{j}\right|=\sum_{q=1}^{2 p+1}\left|a_{j_{q}}\right| \le  \left(2p+1\right)2^{2p} M .
\end{align}
Then for $k \ge 2p$,
\begin{align}
|Q_{k}\left(\gamma\right)|=|\sum_{j=1}^{2 n-1} a_{j} y_j^{k}| \le \sum_{j=1}^{2p+1}|a_{j}|\left(p\tau\right)^{k} \le \left(2p+1\right)2^{2p} M\left(p\tau\right)^{k} .
\end{align}

\textbf{Step 4.} Using (\ref{multi_exp}), we have
\begin{align}
|\max _{\omega \in[-\Omega, \Omega]} \mathcal{F}\gamma\left(\omega\right)| 
& \le \sum_{k \geq 2p}\left(2p+1\right) 2^{2p} M\left(p\tau\right)^{k} \frac{\Omega^{k}}{k!} \notag \\
&<\frac{\left(2p+1\right)2^{2p}M\left(p\tau\Omega\right)^{2p}}{\left(2p\right)!} \sum_{k=0}^{+\infty} \frac{\left(p \tau \Omega\right)^{k}}{k !} \notag \\
&=\frac{\left(2p+1\right)2^{2p}M\left(p\tau\Omega\right)^{2p}}{\left(2p\right)!} e^{p\tau \Omega} \notag \\
& \leq \frac{\left(2p+1\right) M}{2 \sqrt{\pi p}}\left(e \tau \Omega\right)^{2 p} e^{p \tau \Omega} \quad\left(\text {by Lemma \ref{lem_3}}\right) \notag \\
& \leq \frac{\left(2p+1\right) M}{2 \sqrt{\pi p}}\left(e \tau \Omega\right)^{2 p} e^{p}. \quad\left(\text{noting that (\ref{tau}) implies} \,\, \tau \Omega \leq 1\right) .
\end{align}
Finally, using (\ref{tau}) and the inequality that $0.96^{2p}\frac{2p+1}{2\sqrt{\pi p}}\le 1$, we have
\[
	\max _{\omega \in[-\Omega,\Omega]}|\mathcal{F}\gamma\left(\omega\right)|\le \sigma',
\]
which completes the proof.

\subsection{Proof of Theorem \ref{error_upper}}
Let 
$
d = \frac{\pi}{\Omega}\left(\frac{\sigma'}{M}\right),
$
 and $\mathcal{S} = \{z\in[-\frac{\pi}{2\Omega},\frac{\pi}{2\Omega}], |z-y|\ge d\}$.\\
For $\hat{\mu} = a\delta_z$, we have
\begin{align}
    [\hat{\mu}]-[\mu]
    &=\left(\mathcal{F}\hat{\mu}\left(\omega_{-K}\right),\mathcal{F}\hat{\mu}\left(\omega_{-K+1}\right),\cdots,\mathcal{F}\hat{\mu}\left(\omega_K\right)\right)^{\top}-\left(\mathcal{F}\mu\left(\omega_{-K}\right),\mathcal{F}\mu\left(\omega_{-K+1}\right),\cdots,\mathcal{F}\mu\left(\omega_K\right)\right)^{\top}\notag \\
    &=a\left(e^{iz\omega_{-K}},\cdots,e^{iz\omega_{K}}\right)^{\top}-M\left(e^{iy\omega_{-K}},\cdots,e^{iy\omega_{K}}\right)^{\top}.
\end{align}
where $\omega_{-K}=-\Omega$, $\omega_{-K+1}=-\Omega+h$, $\cdots$, $\omega_K=\Omega$ and $h=\frac{\Omega}{K}$.
Then, we have
\begin{align}
    \min_{a\in\mathbb{C},z\in\mathcal{S}}\|[\hat{\mu}]-[\mu]\|_{\infty}
    &=\min_{a\in\mathbb{C},z\in\mathcal{S}}\|a\left(e^{iz\omega_{-K}},1\right)^{\top}-M\left(e^{iy\omega_{-K}},1\right)^{\top}\|_{\infty} \notag \\
    &\ge \frac{1}{\sqrt{2}}\min_{a\in\mathbb{C},z\in\mathcal{S}}\|a\left(e^{iz\omega_{-K}},1\right)^{\top}-M\left(e^{iy\omega_{-K}},1\right)^{\top}\|_{2} \notag\\
    &=\frac{M}{2} \min_{z\in\mathcal{S}}|e^{iy\Omega}-e^{iz\Omega}| \quad(\text{by Theorem 3.1 in \cite{9410626}})\notag\\
    &\ge Md\frac{\Omega}{\pi}\quad(\text{by the inequality that $|e^{i\theta}-e^{i\hat{\theta}}|\ge\frac{2}{\pi}|\theta-\hat{\theta}|$, for $\theta,\hat{\theta}\in[-\frac{\pi}{2},\frac{\pi}{2}]$})\notag\\
    &\ge \sigma'.
\end{align}
Therefore, $|z-y|<d$ is a necessary condition.

\section{Appendix}{\label{sec:appendix}}
\begin{lem}\label{lem_2}
Let $-\frac{\pi}{2}\le\theta_1<\theta_2<\cdots<\theta_n\le\frac{\pi}{2}$ with $\lambda=\min_{i\neq j}|\theta_s-\theta_t|$. Let $\mathcal{S}=\{\theta\in[-\frac{\pi}{2},\frac{\pi}{2}], \min_{j}|\theta-\theta_j|\ge d\}$, where $d\le \frac{\lambda}{2}$. Then, we have
\begin{align*}
    \min_{\theta\in\mathcal{S}} \prod_{j=1}^n |\theta-\theta_j|\ge \frac{\left(n-1\right)!}{2^{n-1}}\cdot d\cdot\lambda^{n-1}.
\end{align*}
\end{lem}

\begin{proof}
    We only prove the case n is even, the case $n$ is odd can be similarly proved.\\
Let $\mathcal{S'}=\{\theta\in[\theta_1,\theta_n], \min_{j}|\theta-\theta_j|\ge d\}$. Clearly, we have
\begin{align*}
    \min_{\theta\in\mathcal{S}} \prod_{j=1}^n |\theta-\theta_j|= \min_{\theta\in\mathcal{S'}} \prod_{j=1}^n |\theta-\theta_j|
\end{align*}
Assume the minimum is achieved at $\theta^*\in\left(\theta_p,\theta_{p+1}\right)$, for some $p$, where $1\le p\le n-1$. Denote $x = \theta^*-\theta_p\in[d,\theta_{p+1}-\theta_p-d]$. We then have
\begin{align*}
    \min_{\theta\in\mathcal{S'}} \prod_{j=1}^n |\theta-\theta_j|
    & =x\left(\lambda+x\right)\left(2\lambda+x\right)\cdots\left(\left(p-1\right)\lambda+x\right)\cdot\left(\lambda-x\right)\cdots\left(\left(n-p\right)\lambda-x\right)\triangleq P\left(x\right).
\end{align*}
It is easy to check that $P\left(x\right)\ge P\left(\theta_{p+1}-\theta_p-d\right)$ when $2p\le n$, and $P\left(x\right)\ge P\left(d\right)$ when $2p\ge n$.\\
For $2p\ge n$, which implies $p\ge \frac{n}{2}$:
\begin{align*}
    P\left(x\right)\ge P\left(d\right)
    &= d\left(\lambda+d\right)\left(2\lambda+d\right)\cdots\left(\left(p-1\right)\lambda+d\right)\cdot\left(\lambda-d\right)\cdots\left(\left(n-p\right)\lambda-d\right)\\
    &\ge d \cdot \left(p-1\right)!\lambda^{p-1} \cdot \frac{\left(2\left(n-p\right)-1\right)!!}{2^{n-p}}\lambda^{n-p}\\
    &\ge d \lambda^{n-1} \frac{\left(2T-2\right)!!\left(2T-1\right)!!}{2^{n-1}}\\
    &\ge d \lambda^{n-1} \frac{\left(n-1\right)!}{2^{n-1}}
\end{align*}
By symmetry, we can derive the same result for $p\le \frac{n}{2}$.
\end{proof}

\begin{lem}{\label{lem_3}}
	For $n\ge 1$, we have
	\begin{align}
		\frac{n^{2n}}{\left(2n\right)!}\le\frac{1}{2\sqrt{\pi}}\left(\frac{e}{2}\right)^{2n}
	\end{align}
\end{lem}
\begin{proof}
	By the following Stirling approximation:
	\begin{align}{\label{stirling}}
	\sqrt{2\pi}n^{n+\frac{1}{2}}e^{-n}\le n! \le en^{n+\frac{1}{2}}e^{-n}.
\end{align}
	we have 
\[
	    	\frac{n^{2n}}{\left(2n\right)!}\le\frac{n^{2n}}{\sqrt{2\pi}\left(2n\right)^{2n+\frac{1}{2}}e^{-2n}}=\frac{1}{2\sqrt{\pi}}\left(\frac{e}{2}\right)^{2n}. 
\]
\end{proof}

\newpage

\bibliographystyle{ieeetr}
\bibliography{references} 
	
\end{document}